\documentclass{amsart}
\usepackage{url}
\usepackage{mathrsfs}
\usepackage{latexsym}
\usepackage{times}
\usepackage{helvet}
\usepackage{courier}
\usepackage{amssymb}
\usepackage{amsmath}
\usepackage{graphicx}
\usepackage{color}
\usepackage{xspace}
\usepackage{bussproofs}
\usepackage{cmll}

\newcommand{\biat}{{\sf E}}
\newcommand{\bel}{{\sf B}}
\newcommand{\ob}{{\sf O}}
\newcommand{\per}{{\sf P}}

\renewcommand{\phi}{\varphi}
\newcommand{\tensor}{\otimes}
\renewcommand{\implies}{\multimap}
\renewcommand\t{\vdash}

\newcommand{\CL}{\ensuremath{\mathsf{CL}}\xspace}

\newcommand{\ILL}{\ensuremath{\mathsf{ILL}}\xspace}

\newcommand{\ILLs}{\ensuremath{\mathsf{ILLs}}\xspace}
\newcommand{\MILLs}{\ensuremath{\mathsf{MILLs}}\xspace}
\newcommand{\AILLs}{\ensuremath{\mathsf{AILLs}}\xspace}
\newcommand{\BILLs}{\ensuremath{\mathsf{BILLs}}\xspace}
\newcommand{\OILLs}{\ensuremath{\mathsf{OILLs}}\xspace}
\newcommand{\ILLD}{\ensuremath{\mathsf{ILLD}}\xspace}

\newcommand{\aILLs}{\ensuremath{\mathsf{aILLs}}\xspace}

\newcommand{\aAILLs}{\ensuremath{\mathsf{aAILLs}}\xspace}
\newcommand{\aBILLs}{\ensuremath{\mathsf{aBILLs}}\xspace}
\newcommand{\aOILLs}{\ensuremath{\mathsf{aOILLs}}\xspace}

\newcommand{\UILLs}{\ensuremath{\mathsf{UILLs}}\xspace}
\newcommand{\XILLs}{\ensuremath{\mathsf{XILLs}}\xspace}
\newcommand{\aXILLs}{\ensuremath{\mathsf{aXILLs}}\xspace}
\newcommand{\BIAT}{\ensuremath{\mathsf{BIAT}}\xspace}

\newcommand{\XCL}{\ensuremath{\mathsf{XCL}}\xspace}
\newcommand{\ACL}{\ensuremath{\mathsf{ACL}}\xspace}
\newcommand{\OCL}{\ensuremath{\mathsf{OCL}}\xspace}
\newcommand{\BCL}{\ensuremath{\mathsf{BCL}}\xspace}
\newcommand{\UCL}{\ensuremath{\mathsf{UCL}}\xspace}
\newcommand\limp\rightarrow
\newcommand\leqv\leftrightarrow

\newtheorem{theorem}{Theorem}
\newtheorem{lemma}[theorem]{Lemma}
\newtheorem{definition}[theorem]{Definition}

\newtheorem{fact}[theorem]{Fact}

\newcommand\AC[1]{\AxiomC{$#1$}}
\newcommand\rl[1]{\RightLabel{#1}}
\newcommand\BC[1]{\BinaryInfC{$#1$}}
\newcommand\UC[1]{\UnaryInfC{$#1$}}
\newcommand\dip{\DisplayProof}

\begin{document}


\title{Logics for modelling collective attitudes}

\address{Free University of Bozen-Bolzano, Piazza Domenicani 3, 39100, Bolzano, Italy. E-mail:
danieleporello{@}gmail.com}

\author{Daniele Porello
} 
\maketitle


\begin{abstract}
We introduce a number of logics to reason about collective propositional attitudes that are defined by means of the majority rule. It is well known that majoritarian aggregation is subject to irrationality, as the results in social choice theory and judgment aggregation show. The proposed logics for modelling collective attitudes are based on a substructural propositional logic that allows for circumventing inconsistent outcomes. Individual and collective propositional attitudes, such as beliefs, desires, obligations, are then modelled by means of minimal modalities to ensure a number of basic principles. In this way, a viable consistent modelling of collective attitudes is obtained.\end{abstract}

\par

\normalsize

\section{Introduction}

The rationality of collective propositional attitudes --- such as beliefs, desires, and intentions --- and of collective agency is a central issue in social choice theory and mathematical economics and it has also become an important topic in knowledge representation and in the foundations of multiagent systems. 
A collective propositional attitude is, generally speaking, a propositional attitude that is ascribed to a collective entity. A map of the most salient notions of collective attitudes was proposed by Christian List \cite{ListErkenntnis2014}, who distinguished between three kinds of collective attitudes: \emph{aggregate}, \emph{common}, and \emph{corporate} attitudes. Corporate attitudes presuppose that the collectivity to which they are ascribed is an \emph{agent} in its own right, an agent who is distinguished from the mere individuals that compose the collectivity. We shall focus, for the main part of this work, on the other two kinds of collective attitudes. Common attitudes are ascribed to collectivities by requiring that every member of the group share the same attitude. Common attitudes have been presupposed for instance by the debate on joint action and collective intentionality \cite{Tuomela2013,ListErkenntnis2014,ListPettit2011}. In this view, possible divergences among the attitudes of the members of the group are excluded and the problem is to understand, for instance, whether an intention that is shared among the members of the group is indeed a collective intention of the group itself.\par
By contrast, aggregative attitudes do not presuppose that every member of the group share the same attitude. In this case, a propositional attitude can be ascribed to the collectivity by solving the possible disagreement by means of a voting procedure such as the majority rule. This view is appealing, since it seems to be capable of accounting for the perspective endowed by common attitudes, for which unanimity is demanded, but also for a number of situations in which it is reasonable to define a collective attitude without assuming that all the members of the group share the same attitudes. For instance, when we model the decisions of parliaments, organizations, or committees, we aim to ascribe collective decisions starting from a situation of initial disagreement. Observe that modelling an aggregative view of collective attitudes is conceptually simpler than modelling common group attitudes in a number of respects. Firstly, we do not need to assume \emph{joint intentionality} nor a shared goal among the group of agents. By the definition of a majoritarian group, we are already assuming that individuals do have different goals and intentions \cite{ListErkenntnis2014}. For that reason, we are terming the attitudes of the group by \emph{group} or \emph{collective} attitudes and not by \emph{joint} attitudes. For an aggregative view of collective attitudes, shared intentionality and shared goals do not need to enter the model for defining what a group attitude is \cite{ListPettit2011}.\par
Besides being descriptively adequate to a number of scenarios, non-unanimous collective attitudes are important also for defining, representing, and assessing group information. Consider the following situations involving artificial agents. Suppose three sensors have been placed in different locations of a room and they are designed to trigger a fire alarm in the case they detect smoke. By viewing the three sensors as a group, we may investigate what are the conditions that define the group action, in this case, ``trigger the alarm'', and its dependency on the group ``beliefs''. By forcing unanimity ---by viewing group attitudes as common attitudes--- we are assuming that the three sensors as a group trigger the alarm only in the case they all agree in detecting smoke. However, a unanimous view of group attitudes may lead to lose important bits of information: if the sensors disagree, the alarm is simply not triggered, even if the disagreement may be caused, for instance, by the fact that one of the three sensors is in a location that has not been reached by the smoke yet. Thus, there are reasons for abandoning common attitudes in modelling information merging \cite{PorelloEndrissJLC2014}. An aggregative view provides the formal means to tailor the concept of collective information to the specific scenario, by selecting the appropriate aggregation method.\par
Although an aggregative view of group attitudes is desirable, several results in social choice theory and judgment aggregation show that many important aggregation procedures are not capable of guaranteeing a rational outcome \cite{ListPuppe2009}. A crucial example is the majority rule, that does not preserve the consistency of individual judgments, as the intriguing discursive dilemmas show \cite{ListPettitEconPhil2002}.
As usual in the (belief-desire-intention) BDI approach to agency, at least a modicum of rationality has to be presupposed in order to define an agent. An agent cannot hold (synchronically) inconsistent attitudes, such as commitments or beliefs. Therefore, when the outcome of an aggregation procedure is inconsistent, as in the case of the majority rule, we simply cannot construe a majoritarian group as an agent. The solution that has been mainly pursued in the literature on judgment aggregation and social choice theory is to give up procedures such as the majority rule and to design aggregation procedures that guarantee consistency \cite{ListPuppe2009}.\par
In this paper, I am interested in pursuing a different strategy: I investigate whether there is a viable notion of rationality with respect to which the outcome of a majoritarian aggregation can be deemed rational. The motivation is that, in several real scenarios, agents actually use the majority rule to settle disagreement. Besides that, the majority rule has a number of desirable features such as it is simple to understand and implement, preference aggregation is non-manipulable (when consistent) \cite{CampbellKelly2003strategy}, it has been associated to a suggestive epistemic virtue justified by the Condorcet's jury theorem \cite{list2001epistemic}.

On a close inspection, as we shall see, the inconsistency of the majority rule is intertwined with the principles of classical logic. In \cite{PorelloIJCAI2013,PorelloPRIMA2015}, a possibility result for the majority rule has been provided by means of a non-classical logic, namely by means of linear logic \cite{Girard1987,Girard1995}. I will build on that in order to develop a number of logics for which majoritarian collective attitudes ---that is, attitudes that are aggregatively defined by means of the majority rule--- are consistent.\par
In principle, this proposal can be applied to any propositional attitudes such as beliefs, desires, intentions, or commitments. For the sake of example, I focus here on three types of propositional attitudes: actions, beliefs and obligations. Moreover, the proposed modelling can in principle be instantiated with a number of aggregation procedures, while I shall focus here on the significant case of the majority rule. \par

The methodology of this paper relies on two approaches. Firstly, at the propositional level, we use an important family of non-classical logics, namely substructural logics ---of which linear logic is an important example--- since they enable a very refined analysis of the collective inconsistency under the majority rule. Secondly, the modelling of individual and collective propositional attitudes makes use of minimal modal logics. This is motivated by the idea of exploring a number of basic principles that govern the reasoning about individual and collective propositional attitudes. A presentation of minimal, or non-normal, modalities is \cite{chellas80}. Non-normal modal logics of actions, beliefs, and obligation that are based on classical propositional logic are extensively studied and discussed in the literature. An exhaustive overview is out of the scope of this paper, I shall only mention a few directly related work in the subsequent sections.\par 
The technical contribution of this paper summarised as follows. For the logical part, I define minimal modal logics of actions, beliefs, and obligations for individual and collective agents based on a significant fragment of substructural logics. I present the Hilbert systems for this logics, their semantics, and I show soundness and completeness. Moreover, I show that the majority rule preserves consistency for interesting fragments of such logics.\par
The conceptual contribution of this paper consists in the application of this logics to propose a consistent modelling of collective attitudes that may serve to the foundation of the status of collective agents.\par
A closely related approach is \cite{BoellaEtAlCOIN2010}, the authors use judgment aggregation for modelling group attitudes by relying on (classical) modal logics of beliefs and goals. The main difference with respect to the present work is that their treatment applies to aggregation procedures that are known to guarantee consistency, e.g. the premise based procedure \cite{ListPuppe2009}, whereas I am interested in approaching the problematic case of the plain majority rule. Another closely related work is the generalisation of the theory of judgment aggregation to capture general propositional attitudes \cite{dietrich2010aggregation}. The present work can be considered a contribution to the study of collective attitudes based on non-classical propositional logics and on weak modalities.\par  

The remainder of this paper is organised as follows. Section 2 discusses a famous case of doctrinal paradox \cite{KornhauserSagerCLR1993} to highlight the problems of a consistent logical modelling of collective attitudes. Section 3 presents the framework of Judgment Aggregation \cite{ListPuppe2009}. In particular, we shall see how the majority rule fails in preserving the notion of consistency based on classical logic. Section 4 introduces and motivates a family of non-classical logics, i.e. substructural logics, as a viable alternative to classical logic to define a notion of consistency that can be preserved by the majority rule at the collective level. Section 5, then, presents a possibility result for the majority rule with respect to substructural logics. Section 6 introduces the basics of the modal logics for modelling individual and collective attitudes. In Section 7, I introduce the proposal for modelling collective attitudes and I exemplify its applications to the doctrinal paradox. In particular, we shall spell out the individual and collective attitudes actually involved in the doctrinal paradox (i.e. beliefs, actions, and obligations). 
Finally, I will approach the treatment of corporate attitudes within the proposed framework, by discussing the nature of the corporate agent who is supposed to be the bearer of the collective attitudes. Section 8 concludes. The technical treatment of the logics introduced in this paper is presented in the Appendices.

\section{The doctrinal paradox}\label{sec:doc}

To illustrate the problems of a theory of collective attitudes, we present the famous case of doctrinal paradox that actually emerged in the deliberative practice of the U.S. Supreme Court, namely the case of \emph{Arizona v Fulminante} \cite{KornhauserSagerCLR1993}. This case originally motivated the study of judgment aggregation, as well as an important debate on the legitimacy of collective decisions, cf. \cite{KornhauserSagerCLR1993} and \cite{ListPettitEconPhil2002,Ottonelli2010}. The Court had to decide whether to revise a trial on the ground of the alleged coercion of the defendant's confession. The legal doctrine prescribes that a trial must be revised if both the the confession was coerced and the confession affected the outcome of the trial.
At the mere level of propositional logic, we formalise the propositions involved as follows: $p$ for ``the confession was coerced'', $q$ for ``the confession affected the outcome of the trial'', and $r$ for ``the trial must be revised''. The legal doctrine is then captured by the formula of classical propositional logic $p \wedge q \rightarrow r$. We only report the votes of three out of the nine Justices of the Supreme Court and we label them by $1$, $2$, and $3$. Individual votes are faithfully exemplified by the following profile.

\begin{center}
\begin{tabular}{ccccccc}
 & $p$ & $p \wedge q$ & $q$ & $p \wedge q \rightarrow r$ & $r$\\
\hline
$1$ & 1 & 1 & 1 & 1 & 1 &\\   
$2$ & 1 & 0 & 0 & 1 & 0 &\\ 
$3$ & 0 & 0 & 1 & 1 & 0 &\\
\hline 
maj. & 1 & 0 & 1 & 1 & 0 &\\
\end{tabular}
\end{center}

By defining the collective attitudes that we ascribe to the Court in an aggregative manner, that is by voting by majority, we obtain the following set of propositions: $p$ is accepted (because of agent 1 and 2), $q$ is accepted (because of 1 and 3), the legal doctrine $p \wedge q \rightarrow r$ is accepted (because it is unanimously accepted), and $r$ is rejected. By viewing the rejection of $r$ as the acceptance of $\neg r$, as usual in this setting, we can easily see that the set of collective attitudes $\{p, q, p \wedge q \rightarrow r, \neg r\}$ is inconsistent. That means that, although each individual set of accepted propositions is consistent, the majority rule does not preserve consistency at the collective level.\par
List and Pettit \cite{ListPettitEconPhil2002} argued that the doctrinal paradox exhibits, besides the irrational outcome, a \emph{dilemma} between a \emph{premise-based} and a \emph{conclusion-based} reading of the majoritarian aggregation. The premise-based reading let the individuals vote on the so called premises $p$ and $q$, then it collectively infers $p \wedge q$ and finally concludes, by $p \wedge q \rightarrow r$, that $r$ is the case.
By contrast, the conclusion-based reading let each individual draw the conclusions by reasoning autonomously on the propositions at issue, then it aggregates the sole conclusions, here $r$, and in this case $r$ is rejected.\par 
For the present applications, three points are worth noticing. 
Firstly, the notion of consistency that is not preserved by the majority rule is the notion of consistency defined with respect to classical logic. It is therefore interesting to investigate whether there are meaningful notions of non-classical consistency that are preserved by the majority rule.\par
Secondly, the distinction between the premise and the conclusion based readings shows that in the doctrinal paradox there are inferences that are performed at the individual level and inferences that are performed at the collective level. In the premise-based reasoning, once $p$, $q$, and $p\wedge q \rightarrow r$ are accepted, the inference that draws $r$ is performed only by a minority of individuals: indeed, this reasoning step is performed only at the collective level on the propositions that have been accepted by majority.
It is then interesting to investigate whether it is possible to make distinction between inferences performed at the individual level and inferences performed at the collective level visible, by means of the logical modelling.\par
Finally, the doctrinal paradox involves a number of individual and collective propositional attitudes such as individual and collective beliefs (concerning whether the confession was coerced) obligations (e.g. the legal doctrine), and actions (e.g. the revision of the trial). For that reason, a proper treatment of the reasoning principles for those attitudes shall be provided. In the remainder of this paper, I will develop a model of collective attitudes that addresses the previous points. 


\section{A model of Judgment Aggregation}\label{sec:ja}

We introduce the basic definitions of the judgement aggregation (JA) setting \cite{ListPuppe2009,EndrissEtAlJAIR2012}, which provides the formal understanding of the aggregative view of collective attitudes \cite{ListErkenntnis2014}. I slightly rephrase the definitions for the present application, by referring to a variable logic $L$ which will be instantiated by a number of logics. Moreover, the definition are presented ``syntactically'', by referring to the derivability relation $\t_{L}$ of a given logic $L$, rather than ``semantically'', as it is usual in JA, in order to focus on the inference principles. In case $L$ is classical logic, the definitions match the standard presentation of JA \cite{ListPuppe2009}. \par


Let $N$ be a (finite) set of agents. Assume that $|N| = n$ and, as usual, $n \geq 3$ and $n$ is odd. An \textit{agenda} $\Phi_{L}$ is a (finite) set of propositions in the language $\mathcal{L}_L$ of a given logic $L$ that is closed under complements, i.e. non-double negations, cf. \cite{EndrissEtAlJAIR2012}. Moreover, we assume that the agenda does not contain tautologies or contradictions, cf. \cite{ListPuppe2009}.\par
A \textit{judgement set} $J$ is a subset of $\Phi_{L}$ such that $J$ is (wrt $L$) \textit{consistent} ($J \nvdash_{L} \bot$), \textit{complete} (for
all $\phi \in \Phi_{L}$, $\phi \in J$ or $\neg \phi \in J$) and \textit{deductively closed} (if $J \t_{L} \phi$ and $\phi \in \Phi_{L}$, $\phi
\in J$). 
Denote by $J(\Phi_{L})$ the set of all judgement sets on $\Phi_{L}$. A \textit{profile} of judgements sets $\textbf{J}$ is a vector $(J_{1}, \dots, J_{n})$, where $n = |N|$. An \textit{aggregator} is then a function $F: J(\Phi_{L})^{n} \to \mathcal{P}(\Phi_{L})$. The codomain of $F$ is the powerset $\mathcal{P}(\Phi_{L})$, therefore we are admitting possibly inconsistent judgments sets. 
 Let $N_{\phi} = \{i \mid \phi \in J_{i} \}$, the majority rule is the function $M: J(\Phi_{L})^{n} \to \mathcal{P}(\Phi_{L})$ such that $M(\textbf{J}) = \{\phi \in \mathcal{X}_{L} \mid |N_{\phi}| > \frac{n+1}{2}\}$.\par
In JA, the collective set $F(\bf{J})$ is also assumed to be consistent, complete, and deductively closed wrt. $L$. We say that the aggregator $F$ is consistent, complete, or deductively closed, if $F(\bf{J})$ is, for every $\textbf{J}$.  The preservation of consistency and completeness with respect to classical logic under any profile defines the standard notion of \emph{collective rationality} in JA. The seminal result in JA can be phrased by saying that the majority rule is not collectively rational \cite{ListPettitEconPhil2002}. That means that there exists an agenda and a profile of judgment sets such that $M(\bf{J})$ is not consistent. A significant example is the doctrinal paradox that we have previously encountered. 
The case of the majority rule generalises to a theorem that applies to any aggregation procedure that satisfies a number of desirable conditions, cf \cite{ListPuppe2009}; that is, many worthy aggregation procedures fail in preserving consistency.\par
We formally define the distinction between the premise-based reading and the conclusion-based reading by defining two aggregation procedures. Suppose that the agenda $\Phi$ is partitioned into two disjoint sets $\Phi^{p}$, the premises of the agenda, and $\Phi^{c}$, its conclusions. The \emph{premise-based} procedure is a function $F$ that firstly aggregates the premises by majority, obtaining the set $S$ of accepted premises, then it infers from $S$ the conclusions $\phi \in \Phi^{c}$; i.e. the output of $F$ is $F(\textbf{J}) \cup \{\phi \mid S \t \phi\}$. The \emph{conclusion-based} procedure simply votes by majority on the conclusions $\phi \in \Phi^{c}$, i.e. those formulas that are inferred individually.\par

\section{Background on Substructural Logics}

We place our analysis in the realm of substructural logics \cite{Paoli2002,Restall2002} to show that the majority rule is in this case quite surprisingly consistent. Substructural logics are a family of logics that reject the (global) validity of the classical structural principle of contraction (C) and weakening (W).\footnote{The principles are labelled \emph{structural} as they correspond to the structural rules of the classical sequent calculus \cite{Paoli2002}. Another structural principle is the \emph{exchange}, which entails the commutativity of the conjunction and of the disjunction.} This principles are captured in classical logic by the axioms (W) $\phi \rightarrow (\psi \rightarrow \phi)$ and (C) $(\phi \rightarrow (\phi \rightarrow \psi)) \rightarrow (\phi \rightarrow \psi)$. Rejecting (W) amounts to preventing the monotonicity of the entailment, while rejecting (C) blocks the possibility of identifying several copies of the assumptions when drawing inferences.

An important family of substructural logics is that of \emph{relevant logics} which are traditionally motivated in philosophical logic by the aim of capturing informative inferences \cite{MasoloPorelloAIC2015,Mares2004}.  Another important family of substructural logics is that of \emph{linear logics}, which are motivated mainly as logics of computation \cite{Girard1995, Abramsky1993}. The linear implication, denoted by $\implies$, keeps track of the amount of formulas actually used in the deduction: for instance,  \emph{modus ponens} $\phi, \phi \implies \psi \t \psi$ is valid only if the right amount of assumptions is given, so that $\phi, \phi, \phi \implies \psi \not\t \psi$. Resource-sensitivity is thus achieved by rejecting (C) and preventing the identification of two occurrences of $\phi$. 
The resource-sensitivity of this logics has been used in applications to a number of topics in knowledge representation and multiagent systems such as planning \cite{KanovichVauzeilles2001}, preference representation and resource allocation \cite{HarlandWinikoffAAMAS2002, PorelloEndrissKR2010,PorelloEndrissECAI2010}, social choice theory \cite{PorelloIJCAI2013}, action modelling \cite{BorgoEtAlFOIS2014,PorelloTroquardJANCL2015}.\par

A crucial observation that motivates the use of substructural logic for modelling collective attitudes is that, by rejecting (W) and (C), we are led to split the classical connectives into two classes: the \emph{multiplicatives} and the \emph{additives}.\footnote{In the tradition of relevant logics, the distinction between multiplicatives and additives corresponds, respectively, to the distinction between \emph{intensional} and \emph{extensional} connectives, cf. \cite{Paoli2002}.} For instance, the classical conjunction $\wedge$ splits into two distinct operators: the multiplicative $\otimes$ (``tensor'') and the additive $\with$ (``with'') with distinct operational meaning \cite{Girard1987,Girard1995}. The inferential behaviour of the two conjunctions in fact differs: one can prove that $\gamma \implies  \phi, \delta \implies \psi \t \gamma \otimes \delta \implies \phi \otimes \psi$  and that $\gamma \implies \phi, \gamma \implies  \psi \t \gamma \implies \phi \with \psi$, that is, the tensor case demands the combination of hypotheses $\gamma \otimes \delta$, whereas  the with case requires that the same hypothesis $\gamma$ has been used to infer $\phi$ and $\psi$. 
In presence of (W) and (C), $\otimes$ and $\with$ are provably equivalent, therefore they are indiscernible in classical logic. For this reason, the classical conjunction is more powerful and permits both inferential patterns. Analogous distinction can be made for disjunctions.\par
For the purpose of this paper, the distinction between multiplicatives and additives is curiously related to an important distinction between the truth makers of a proposition in an aggregative setting: we will see that a multiplicative formula $\phi \otimes \psi$ is made true by two possibly different coalitions of agents, one that supports $\phi$ and one that supports $\psi$, whereas $\phi \with \psi$ shall be made true by a single coalition of agents that supports both $\phi$ and $\psi$, cf. \cite{PorelloIJCAI2013}. This distinction is crucial for the preservation of consistency under the majority rule and we shall use it for distinguishing the reasoning steps that are performed at the individual level from those performed at the collective level.\par
The basic logic that we are going to use throughout this paper is the intuitionistic version of multiplicative additive linear logic.\footnote{Linear logic is capable of retrieving the structural rules in a controlled manner by means of the exponential operators \cite{Girard1987,Girard1995}. Here we confine ourselves to the fragment of multiplicative additive intuitionistic linear logic (exponential-free). We also leave a proper comparison with the families of substructural and relevant logics for future work.}
The motivation for using the \emph{intuitionistic} version is mainly technical: intuitionistic linear logic allows for an easy Kripke-like semantics which in turns provides an easy way to define the semantics of modalities. In principle, the analysis of collective attitudes of this article can be performed by means of the classical version of linear logic. In particular, intuitionistic and classical logic do no differ with respect to the aggregation of propositions under the majority rule \cite{PorelloECSI2014}. For this reason, endorsing the intuitionistic restriction does not entail a substantial loss of generality of the approach in this case. We preferred, for the sake of simplification, to stick to the intuitionistic version of linear logic, leaving a proper definition of the semantics of modalities for classical linear logic for a future dedicated work.\footnote{The semantics of classical linear logic is provided by means of \emph{phase spaces}, \cite{Troelstra1992}. In order to define modalities for classical linear logic, we have to extend phase spaces by introducing neighbourhood functions on those structures and adapting the completeness proof for this case. This requires a dedicated work.}
Moreover, another possibility is to work with a fragment of \emph{relevant logic} which admits the same type of semantics \cite{Paoli2002,Restall2002}.\par
 The only addition to intuitionistic linear logic is the adjunction of a \emph{strong} negation $\sim$ which is motivated for modelling agents' acceptance and rejection of propositions in the agenda in a symmetric manner, as it is usual in the setting of JA.\footnote{The intuitionistic negation, which is definable by means of the implication and the symbol for absurdity (i.e. $A \implies \bot$), does not capture the intended meaning of the rejection of an item in the agenda.}
Linear logic with strong negation has been studied in particular in \cite{wansing1993informational}. We term this system by $\ILLs$.

The language of $\ILLs$, $\mathcal{L}_{\ILLs}$, then is defined as follows. Assume a set of propositional atoms $Atom$, $p \in Atom$, then:\footnote{We use the notation introduced by Girard \cite{Girard1987}, denoting the additive conjunction and disjunction by $\with$ and $\oplus$ respectively. In \cite{wansing1993informational} and \cite{kamide2006linear}, additives are denoted by means of the classical notations $\wedge$ and $\vee$.} 

\[\phi ::= p \mid \textbf{1} \mid \top \mid \bot \mid \sim \phi \mid \phi \otimes \phi \mid \phi \implies \phi \mid \phi \with \phi \mid \phi \oplus \phi\] 

By splitting connectives into additives and multiplicatives, the units of the logic split as well: we have a multiplicative constant for truth $\textbf{1}$, which is the neutral element for $\otimes$, and an additive constant for truth $\top$, which is neutral for $\with$. In the intuitionistic case, we only have one constant for falsity $\bot$ \cite{wansing1993informational,kamide2006linear}. 

\subsection{Hilbert system for \ILLs}
\label{sec:hilbert}

The Hilbert system for intuitionistic (multiplicative additive) linear logic has been basically developed in \cite{AvronTCS1988}, see also \cite{Troelstra1992,Paoli2002}. We extend \ILL to \ILLs by adding a strong negation as in \cite{wansing1993informational}.
We define the Hilbert-style calculus by introducing a list of axioms in Table 1 and by defining the following notion of deduction. The concept of deduction of linear logic requires a tree-structure in order to handle the hypothesis in the correct resource-sensitive way. Notice that, by dropping weakening and contraction, we have to consider \emph{multisets} of occurrences of formulas. This entails that, in particular, in linear logic, every \emph{modus ponens} application (cf. $\implies$-rule) applies to a single occurrence of $\phi$ and of $\phi \implies \psi$.  

\begin{table}
\label{hills}
\begin{enumerate} \setlength\itemsep{0.5pt}
\item $\phi \implies \top$
\item $\bot \implies \phi$
\item $\phi \implies \phi$
\item $ (\phi \implies \psi) \implies ((\psi \implies \gamma) \implies (\phi \implies \gamma))$
\item $ (\phi \implies (\psi \implies \gamma)) \implies (\psi \implies (\phi \implies \gamma))$
\item $ \phi \implies (\psi \implies \phi \otimes \psi)$
\item $ (\phi \implies (\psi \implies \gamma)) \implies (\phi \otimes \psi \implies \gamma)$
\item $ \textbf{1}$
\item $ \textbf{1} \implies (\phi \implies \phi)$
\item $ (\phi \with \psi) \implies \phi$
\item $ (\phi \with \psi) \implies \psi$
\item $ ((\phi \implies \psi) \with (\phi \implies \gamma)) \implies (\phi \implies \psi \with \gamma)$
\item $ \phi \implies \phi \oplus \psi$
\item $ \psi \implies \phi \oplus \psi$
\item $ (\phi \implies \gamma) \with (\psi \implies \gamma) \implies (\phi \oplus \psi \implies \gamma)$
\item $ \phi \implies \sim \sim \phi$
\item $ \sim \sim \phi \implies \phi$
\item $ \sim \textbf{1} \implies \phi$
\item $ \sim \top \implies \phi$
\item $ \phi \implies \sim \bot$
\item $ \sim (\phi \implies \psi) \implies~ \sim \phi \otimes \psi$
\item $\phi \otimes \sim \psi \implies~ \sim (\phi \implies \psi)$
\item $\sim (\phi \with \psi) \implies~ \sim \phi~\oplus~\sim \psi$
\item $\sim \phi \oplus \sim \psi \implies~ \sim (\phi \with \psi)$
\item $\sim (\phi \oplus \psi) \implies~ \sim \phi~\with~\sim \psi$
\item $\sim \phi~ \with \sim \psi~ \implies~ \sim (\phi \oplus \psi)$
\item $\sim (\phi \otimes \psi)~ \implies~  \sim \phi~\otimes~\sim \psi$
\item $\sim \phi~\otimes~\sim \psi ~ \implies~ \sim(\phi \otimes \psi)$  
\end{enumerate}
\caption{Axioms of \ILLs}
\end{table}

The notion of proof in the Hilbert system is then defined as follows.

\begin{definition}[Deduction in \ILLs]\label{def:derivation-ill}
A \emph{deduction tree} $\mathcal{D}$ in \ILLs  is inductively constructed as follows. (i)~The leaves of the tree are assumptions $\phi \t \phi $, for $\phi \in \mathcal{L}_{\ILLs}$, or
$\t \psi$ where $\psi$ is an axiom in Table 1 (base cases).\\
 (ii) We denote by $\stackrel{\mathcal{D}}{\Gamma \t \phi}$ a deduction tree with conclusion $\Gamma \t  \phi$. If $\mathcal{D}$ and $\mathcal{D}'$ are deduction trees, then the following are deduction trees (inductive steps). 

\begin{center}
\begin{tabular}{ccc}
$
\AC{\stackrel{\mathcal{D}}{\Gamma \t  \phi}}
\AC{\stackrel{\mathcal{D}'}{\Gamma' \t  \phi \implies \psi}} \rl{$\implies$-rule}
 \BC{\Gamma, \Gamma' \t  \psi}
 \dip
$

&

\hspace{10pt}
& 

$
\AC{\stackrel{\mathcal{D}}{\Gamma \t \phi}}
\AC{\stackrel{\mathcal{D}'}{\Gamma \t \psi}} \rl{$\with$-rule}
\BC{\Gamma \t \phi \with \psi}
\dip
$
\\
\end{tabular}
\end{center}

\end{definition}

We say that $\phi$ is derivable from the multiset $\Gamma$ in \ILLs,  denoted by $\Gamma \t \phi$, iff there exists a deduction tree with conclusion $\Gamma \t \phi$. The semantics of \ILLs is presented in Appendix \ref{app:sem}.

\section{A possibility result for substructural logics}\label{sec:possibility}

By assessing the outcome of the majority rule with respect to classical logic, collective rationality fails. In \ILLs the situation radically changes. Recall that the \emph{multiplicative connectives} are $\otimes$ and $\implies$ and the \emph{additive connectives} of \ILLs are $\with$ and $\oplus$. We denote the additive implication by $\phi \rightsquigarrow \psi$ and we define it by $\sim \phi \oplus \psi$. We label by \aILLs the additive fragment of \ILLs whose Hilbert system is obtained by taking all the axioms that contain additive formulas plus the $\with$-rule (therefore we drop the $\implies$-rule).\par
If $\phi$ is a formula in classical logic, we define its \emph{additive translation} $\phi'$ as follows: $p' = p$, for $p$ atomic, $(\neg \phi)' =\; \sim \phi'$, $(\phi \wedge \psi)' = \phi' \with \psi'$ and $(\phi \vee \psi)' = \phi' \oplus \psi'$. The additive translation of the outcome of the majority rule is then simply $M(\textbf{J})' = \{\phi' | \phi \in M(\textbf{J})\}$.\par

As it is standard in JA, we assume that the individuals reason by means of classical logic (\CL), however we assess the outcome of the majority rule with respect to the logic \aILLs, by means of the additive translation. This returns a voting procedure from profiles of judgments defined on a certain classical agenda $\Phi_{\CL}$ to sets of judgments defined in \aILLs.\footnote{Observe that Theorem \ref{th:aills} does not depend on the logic being intuitionistic, it holds for a classical version of additive linear logic \cite{PorelloIJCAI2013}. We present it for \aILLs, since it is the logical setting of this paper.}

\begin{theorem}\label{th:aills}
For any agenda $\Phi_{\CL}$ in classical logic, suppose that for each $i \in N$, $J_{i}$ is consistent and complete wrt classical logic, then for any profile $\textbf{J}$, $M(\textbf{J})'$ is consistent and complete wrt \aILLs. 
\end{theorem}

\begin{proof}
The proof is based on the fact that, in additive linear logic, every minimally inconsistent  set has cardinality 2.\footnote{Recall that a minimally inconsistent set $Y$ is an inconsistent set that does not contain inconsistent subsets.} This follows from the fact that every deduction in the additive fragment of linear logic contains exactly two formulas $\phi, \psi$, as it has been noticed in \cite{HughesGlabbeek2003}.\footnote{This condition surprisingly corresponds to the \emph{median property} introduced by \cite{NehringPuppeJET2007}. See also \cite{ListPuppe2009}. This condition in fact characterizes agendas for which the majority rule returns consistent sets of judgments.}
Suppose then by contradiction that $M(\textbf{J})$ is inconsistent. Then, there exists a set of two formulas $\{\phi,\psi\} \subseteq M(\textbf{J})$ such that $\{\phi, \psi\}$ is inconsistent. Since both $\phi$ and $\psi$ are in  $M(\textbf{J})$, this entails that $|N_{\phi}| \geq \frac{n+1}{2}$ and $|N_{\psi}| \geq \frac{n+1}{2}$, which entails that there exists an individual $i$ such that $\phi$ and $\psi$ are in $J_{i}$. If $\{\phi, \psi\}$ is inconsistent wrt \aILLs, then $\{\phi, \psi\}\t \bot$, therefore $\{\phi, \psi\}$ is also inconsistent wrt  classical logic, since classical logic derives more theorems than \aILLs. 
Therefore, $J_{i}$ is inconsistent, against the assumption of (classical) consistency of individual judgment sets. Completeness of $M(\textbf{J})'$ follows by noticing that the majority rule always preserves completeness \cite{EndrissEtAlJAIR2012}. 
\end{proof}

Theorem \ref{th:aills} shows that, although the majority rule does not preserve the notion of consistency of classical logic, it is still capable of preserving the notion of consistency defined by means of \aILLs.\footnote{A number of extensions of \aILLs for which consistency is preserved by majority is discussed in \cite{PorelloIJCAI2013}.}\par
Observe that the notion of consistency of \aILLs is \emph{weaker} than the notion of consistency of classical logic. Since in classical logic we can prove more theorems than in \aILLs, if a set of formulas is consistent in classical logic --- i.e. we cannot prove $\bot$--- then its translation in \aILLs is consistent. By contrast, if a set of formulas is consistent in \aILLs ---i.e. it cannot prove $\bot$--- that does not mean that its classical counterpart is consistent, since classical logic is more powerful.\par

Take the case of the doctrinal paradox that we have previously discussed. The agenda of propositions in classical logic is given by the following set: $\{p, q, p \wedge q,  (p \wedge q) \rightarrow r, r, \neg p, \neg q, \neg (p \wedge q), \neg ((p \wedge q) \rightarrow r), \neg r\}$. The salient part of the profile of judgments $\textbf{J}$ involved in the doctrinal paradox is represented in Table 2.

\begin{table}\label{dd}
\begin{center}
\begin{tabular}{cccccccccc}
 & $p$ & $p \wedge q$ & $q$ & $\neg p$ & $\neg (p \wedge q)$ & $\neg q$ & $ (p \wedge q) \rightarrow r $ & $r$ & $\neg r$ \\ 
 \hline
 1 & yes & yes & yes & no & no & no & yes & yes & no\\
 2 &  no & no & yes  & yes & yes & no & yes & no & yes\\ 
 3 & yes & no & no  & no & yes & yes & yes & no & yes\\   
 \hline
 maj. & yes & no & yes  & no & yes & no & yes & no & yes\\
\end{tabular}
\end{center}
\caption{The doctrinal paradox}
\end{table}

In classical logic, the outcome of the majority rule $M(\textbf{J})$ is inconsistent since, for instance, the acceptance of $p$ and of $q$ entails the acceptance of $p \wedge q$ which contradicts $\neg (p \wedge q)$ (i.e. $p, q, \neg (p \wedge q) \t \bot$). Moreover, the premise-based reading and the conclusion-based reading provide contradictory outcomes. The premise-based reading accepts $p$, $q$, and $p \wedge q \rightarrow r$, then infers $r$; by contrast, the conclusion-based procedure just votes on the conclusion ($r$ and $\neg r$) and accepts $\neg r$.\par
The additive translation of $M(\textbf{J})$ is $\{p, q, \sim (p \with q), (p \with q) \rightsquigarrow r, \sim r\}$. Firstly, notice that in \aILLs, we cannot infer the \emph{additive} conjunction $p \with q$ from $p$ and $q$, we could only infer the multiplicative conjunction $p \otimes q$ (by Axiom 6). Thus, the (multi)set $p, q, \sim (p \with q)$ is not inconsistent in \aILLs.\footnote{In \ILLs $p \otimes q, \sim (p \with q) \not\t \bot$, therefore, in \aILLs, $p,q, \neg (p \wedge q) \not\t \bot$.}\par
Moreover, the premise-based and the conclusion-based reading do not provide contradictory outcomes. The conclusion-based reading provides $\sim r$, whereas the premise-based reading does not infer any conclusion now, since $p, q, p \with q \rightsquigarrow r \not\t r$ in \ILLs.
Therefore, the analysis of the doctrinal paradox in terms of \aILLs takes the conclusion-based horn of the dilemma. The plain majority rule provides the same outcome as the conclusion-based procedure.\footnote{Two arguments are usually raised against the conclusion-based procedure: Firstly, it is usually not complete, secondly the reasons of the collective decisions, i.e. the premises, are invisible in the collective set. 
For a definition of complete conclusion-based procedure, see \cite{pigozzi2009complete}. The solution provided by means of substructural logics shows that the plain majority rule provides the conclusion-based solution, therefore, it is both complete and it shows in the collective sets the reason for the collective choice.}\par

We see how the distinction between additives and multiplicatives of substructural logics allows for distinguishing inferences performed at the individual level and inferences performed at the collective level. Firstly, notice that since the Hilbert system of \ILLs is sound and $M(\textbf{J})'$ is consistent, the following fact holds.

\begin{fact}\label{fact} The deductive closure in \ILLs of $M(\textbf{J})'$, $Cl(M(\textbf{J})'$), is consistent.
\end{fact}

\noindent
Henceforth, we can safely reason about $M(\textbf{J})'$ by means of full \ILLs.  
By Fact \ref{fact}, we can extend the majority rule $M(\textbf{J})'$, which is defined from profile in $J(\Phi_{\CL})^{n}$ to subsets of an agenda in \aILLs, to any agenda $\Phi_{\ILLs}$ of formulas of full \ILLs that contains the additive translation of the classical agenda $\Phi_{\CL}$.

\begin{equation}\label{eq:cl_maj}
Cl_{\ILLs}(M(\textbf{J})') = \{ \phi \in \Phi_{\ILLs} \mid M(\textbf{J})' \t \phi\}
\end{equation}

\noindent
By means of the additive translation, we are viewing the formulas that have been accepted by majority as additive formulas. That is, the inferences that have been performed autonomously by (a majority of) the individuals are visible in $Cl_{\ILLs}(M(\textbf{J})')$ as additive formulas. Multiplicatives then can only combine formulas that have already been accepted by majority. For this reason, we view multiplicatives as performing inferences at the collective level.\par
For instance, if we assume that $\Phi_{\ILLs}$ contains the formula $p \otimes q$, we can infer from the acceptance of $p$ and $q$ in $M(\textbf{J})'$, the multiplicative conjunction $p \otimes q$. The difference between $p \otimes q$ and $p \with q$ in this setting is that $p \with q$ requires a single majority of agents that support $p$ and $q$, in order to be collectively accepted. By contrast, $p \otimes q$ only requires that there exist two possibly distinct majorities of individuals, one that supports $p$ and one that supports $q$.\par
The reason why the analysis of the doctrinal paradox by means of substructural logics takes the conclusion-based solution is that the legal doctrine has been interpreted in additive terms by means of the formula $p \with q \rightsquigarrow r$. This formulation entails that $r$ is accepted only if a majority of individuals autonomously infer $r$.

\subsection{Collective reasoning}

In order to obtain the solution provided by the premise-based reading of the doctrinal paradox, as we argued, we have to introduce a reasoning step that is not performed by (a majority of) the individuals, that means a reasoning step that takes place only at the collective level. For this reason, the premise-based reading requires a multiplicative formulation. We can obtain the premise-based solution to the dilemma by dropping the additive formulation of the legal doctrine and replace it with a multiplicative formulation: $p \otimes q \implies r$.  Notice that, by keeping both the formulations of the legal doctrine, the inconsistency between accepting and rejecting $r$ is back. This fact shows why in classical logic, that does not make the distinction between multiplicatives and additives, the dilemma between the premise-based and the conclusion-based procedures is challenging and corresponds to the inconsistency of the majority rule.\par
Since the multiplicative formulation of the legal doctrine applies as step of reasoning performed only at the collective level, it has to be modelled as a constraint that is external to the agenda of propositions about which the agents express their votes.\footnote{The view of constraints that are external to the agenda is close to the approach in \cite{GrandiEndrissAI2013}.} 
We term this external constraints \emph{collective constraint} and we define a set $\Sigma$ of collective constraints as a set of formulas in $\mathcal{L}_{\ILLs}$. A collective constraint norms the way in which the group has to reason once a number of propositions is accepted by majority. Collective constraints do not belong to the individual judgments sets and they are not used to perform inferences at the individual level. 
We can in principle assume that collective constraints are subject to aggregation, namely, that the individuals decide by voting how the group is supposed to reason about the collectively accepted propositions. To model that, we assume that the individual judgments are partitioned into two sets: a subset $J_i$ of the agenda $\Phi_\CL$ (as usual) and a subset $S_i \subseteq \Sigma$ of collective constraints. We assume that, for each $i$, $S_i$ is consistent wrt \ILLs; moreover, the collective constraints are compatible with the judgments of the individuals, once they are interpreted as additive formulas: that is, we assume that $(J_i)' \cup S_i$ is consistent in \ILLs. 
Denote by $\textbf{S}$ the profile $(S_1, \dots, S_N)$ of individually accepted collective constraints. The majority rule $M(\textbf{S})$ is defined as usual from profiles of constraints to sets of constraints in $\Sigma$.\footnote{We are proposing to model the decision about the collective constraints by means of the majority rule, in order to meet the standard view of the doctrinal paradox, where the legal doctrine was in fact a matter of voting.} We define the following \emph{two steps} procedure.

\begin{definition}\label{def:twostep}
Let $\Phi_{\CL}$ an agenda in \CL and $\Phi_{\ILLs}$ an agenda in \ILLs that is composed of three disjoint sets: the additive translation of $\Phi_{\CL}$, a set of collective constraint $\Sigma$, and a set of conclusions $\Delta$. The two step procedure $T: J(\Phi_{\CL})^{n} \to \mathcal{P}(\Phi_{\ILLs})$ is defined as follows:

$$T(\textbf{J}) = \{\phi \mid \phi \in M(\textbf{J})'\; \text{or}\; \phi \in M(\textbf{S})\; \text{or}\; \phi \in \Delta\;\text{and}\; M(\textbf{J})', M(\textbf{S}) \t \phi \}$$
\end{definition}

\noindent
$T$ operates as follows: firstly, $T$ aggregates by majority the formulas in the individuals' agenda $\Phi_{\CL}$ (and translates them into additive formulas) and the collective constraints in $\Sigma$; secondly, $T$ infers possible conclusions in $\Delta$, by means of the logic \ILLs and the constraints in $M(\textbf{S})$.\footnote{Note that by construction, $T$ returns a set of formulas, not a multiset. $T(\textbf{J})$  is nonetheless evaluated wrt \ILLs by viewing it as a multiset with multiplicity 1. Although multisets have no straightforward interpretation in an aggregative setting, it is the device that is necessary to really drop contraction. Once we assume that formulas form sets, then $\{\phi \}$ and $\{\phi, \phi\}$ become indistinguishable and that would enable contraction, which entails $A \with B \implies A \otimes B$.
In principle, contraction is not harmful for the majoritarian aggregation, cf. \cite{PorelloIJCAI2013}. Thus we could add contraction and work with sets, rather then with multisets in this case. We preferred to assume multisets here, because they provide the standard means to define the Hilbert systems for linear logics.
We leave a treatment for logics enabling contraction in an aggregative setting for future work.}\par

%

If we do not assume any collective constraint,  i.e. $\Sigma = \emptyset$, then the procedure of Definition \ref{def:twostep} coincide with the procedure defined in Equation \ref{eq:cl_maj}. Thus, in that case, the consistency of $T$ is guaranteed. Unfortunately, by adding possibly arbitrary constraints in full \ILLs, consistency is not ensured any longer.\footnote{This is standard for premise-based procedures in the case individuals are permitted to vote on non-atomic propositions.} To guarantee the consistency of $T$ in case the constraints $\Sigma$ are expressed by means of any multiplicative formula, we have to assume that the following property holds: every subset of $(\Phi_{\CL})' \cup \Sigma$ that is minimally inconsistent in \ILLs has cardinality at most 2.\footnote{This is the median property instantiated for the logic \ILLs, that provides necessary and sufficient conditions for the consistency of the majority rule, cf. \cite{PorelloIJCAI2013}. See also \cite{PorelloEndrissJLC2014} for the use of the median property under external constraints.}\par
The premise-based reading of the doctrinal paradox can be obtained by instantiating $T$ as follows: $\Phi_{\CL} = \{p, q, \neg p, \neg q, p \wedge q, \neg (p \wedge q)\}$,  $\Sigma = \{ p \otimes q \implies r\}$, and $\Delta = \{r, \sim r, p \otimes q\}$.\footnote{Note that in this case $(\Phi_{\CL})' \cup \{ p \otimes q \implies r \}$ satisfies the median property.} By assuming the same profile as in Table 2 on the judgments $J_i$ and by assuming unanimity on the collective constraint, the outcome of $T$ is then $\{p, q, \sim (p \with q), p \otimes q, p\otimes q \implies r, r\}$. The collective set shows that $p \with q$ does not hold because there is no single majority of agents that supports both $p$ and $q$, $p \otimes q$ is accepted because there are indeed two majorities that support $p$ and $q$, and that is sufficient to apply the collective constraint $p \otimes q \implies r$ in $\Sigma$, and conclude by \emph{modus ponens} $r$.\par
It is possible to obtain the conclusion-based outcome by means of $T$ by simply replacing $\Sigma$ with $\Sigma' = \{p \with q \implies r\}$, that imposes the additive formulation of the legal doctrine.\footnote{Observe that in presence of weakening and contraction, as in classical logic, the two formulations of the legal doctrine are equivalent.} In this view, the additive formulation forces the conclusion-based reading, whereas the multiplicative constraint entails the premise-based reading. From this perspective, the additive formulation is more demanding in terms of the cohesion of the group: it requires that the \emph{same} majority accepts the premises in order to draw the conclusion, whereas the premise-based reading only requires that there exist possibly distinct majorities on all the premises.\par

 \ILLs provides an interesting propositional base to model collective attitudes. In the next Section, we introduce a number of modal extensions of \ILLs in order to model the collective propositional attitudes involved in the doctrinal paradox.

\section{Modal logics for modelling collective attitudes}

We present three modal extensions of \ILLs to reason about actions, beliefs, and obligations of individuals and collectives. We assume a set of coalitions $\textbf{C}$ and to express individual propositional attitudes, we admit singleton coalitions; in that case the meaning of a coalition $C$ in $\textbf{C}$ is $\{ i \}$.\footnote{This move is similar to the approach in \cite{TroquardJAAMAS2014} to discuss coalitional ability.}
We introduce three modalities $\biat_C \phi$, $\bel_C \phi$, and $\ob_C \phi$ to express that coalition $C$ \emph{does} $\phi$, \emph{believes} that$\phi$, and \emph{is obliged to} $\phi$. We label this logics by \AILLs, \BILLs, and \OILLs, for action, beliefs, and obligations, respectively.\par
We use minimal (or non-normal) modalities in order to ensure a number of basic principles to reason about this modalities. The discussion of further principles is left for  future work.\par
As we have seen in the previous section, in order to understand and possibly circumvent collective inconsistency, we crucially rely on the distinction between multiplicatives and additives. The role of modalities is to make the coalitions of agents that support a certain proposition explicit. 
Let $\Box$ be one of the modalities in $\{\biat, \bel, \ob\}$. The distinction between multiplicatives and additives affects the modal extension in the so called \emph{combination axioms}, which is usually admitted in the classical counterparts of this logics: $\Box_C \phi \wedge \Box_D \psi \rightarrow \Box_{C \cup D} (\phi \wedge \psi)$. 
This axiom can be interpreted in substructural logics in two ways: a multiplicative and an additive way. The multiplicative combination of two propositions, in our interpretation, requires two possibly different winning coalitions that support each combined propositions, whereas the additive combination forces the same coalition to support both propositions. This distinction is reflected by the operational meanings of the two conjunctions.\footnote{This distinction is related to the distinction between group knowledge and distributed knowledge in epistemic logics, cf. for instance \cite{aagotnes2016coalition}.} Therefore, the additive version of the axiom means that if the same coalition support $\phi$ and $\psi$, then the same coalition supports the additive combination of $\phi$ and $\psi$: $\phi \with \psi$. The multiplicative means that if a coalition $C$ supports $\phi$ and coalition $C'$ supports $\psi$ then, the \emph{disjoint union} of the two coalitions $C \sqcup C'$ supports the combination $\phi \otimes \psi$. The condition of disjointness of $C$ and $C'$ is crucial for modelling the collective attitudes in a consistent way. In particular, the condition shows that the individuals that are member of the coalition are all equally relevant to make the propositions accepted. Take for instance the case of $\Box_{\{1,2\}} \phi$ and $\Box_{\{2,3\}} \psi$: if we enable the inference to $\Box_{\{1,2,3\}} \phi \otimes \psi$, then we lose the information concerning the marginal contribution of agent 2 in both coalitions. For this reason, we assume that the set \textbf{C} of coalitions is closed under disjoint union.\footnote{We define the disjoint union of two coalitions $C \sqcup D$ by $C \cup D$, if $C \cap D = \emptyset$ and $C \sqcup D = (C \times \{1\}) \cup (D \times \{0\})$, otherwise.}
The combination axiom is reminiscence of coalition logic \cite{PaulyCoalitionLogic2002}. Note that we do not assume any further axiom of coalition logic. For instance, we do not assume \emph{coalition monotonicity}: $\Box_C \phi \implies \Box_D \phi$, if $C \subseteq D$. The motivation is that we are modelling \emph{profile-reasoning}, that is, we start by a fixed profile of individual attitudes and we want to capture, by means of the modalities $\Box_C$, how the collectivity reasons about the propositions that have been accepted by majority in that profile. In this setting, given a profile of individual attitudes, there exists only one coalition that supports a proposition that has been accepted by majority. This is a different perspective wrt coalition logic and logic of coalitional ability \cite{TroquardJAAMAS2014}.
The technicalities of this logics are presented in Appendix \ref{app:mills} and \ref{app:abo}.\footnote{We do not discuss here the computational complexity of the logics that we introduce. We simply summarise a few points. Intuitionistic linear logic is PSPACE complete \cite{KanovichVauzeilles2001} and the logic of action based on \ILL is proved to be in PSPACE in \cite{PorelloTroquardJANCL2015}. For the case of  \ILLs (i.e. with strong negation), we have decidability, since the version of \ILLs based on sequent calculus enjoys cut elimination \cite{wansing1993informational}. A number of normal modal logics  based on \ILLs are also proved to be decidable \cite{kamide2006linear}. Therefore, in principle, it is possible to show that also our modal logics are decidable. Moreover, it is known that the fusion of decidable logics is decidable \cite{blackburn2006handbook}, thus we can in principle adapt the standard theory for the case of non-normal modal logics based on \ILLs. 
}

\subsection{Collective actions}

The logic to reason about collective actions is based on the minimal logic of \emph{bringing-it-about}, which was traditionally developed on top of classical propositional logic~\cite{Elgesem97agency,Governatori05Elgesem}. The principles of this logic aim to capture a very weak view of actions that, for instance, does not presuppose intentionality or explicit goals. We apply this minimal view to conceptualise collective actions. This is adequate for an aggregative perspective on collective actions, for which the collective is not assumed, in general, to have joint intentionality nor any shared goal, \cite{ListErkenntnis2014}.\par

 In \cite{PorelloTroquardECAI2014}, a version of bringing-it-about based on \ILL  was introduced for modelling resource-sensitive actions of a single agent. We propose here a version with coalitions, based on \ILLs, we label it by \AILLs (action-\ILLs). Four principles of agency are captured by the classical bringing-it-about logic~\cite{Elgesem97agency}. The first  corresponds to the axiom $T$ of modal logics: $\biat_i \phi \implies \phi$, it states that if an action is brought about, then the action affects the state of the world, i.e. the formula $\phi$ that represents the effects of the action holds. The second principle  corresponds to the axiom $\lnot \biat_i \top$ in classical bringing-it-about logic. It amounts to assuming that agents cannot bring about tautologies. The motivation is that a tautology is always true, regardless what an agent does, so if acting is construed as something that affects the state of the world, tautologies are not apt to be the content of something that an agent actually does. 
The third principle corresponds to the axiom: $\biat_i \phi \land \biat_i \psi \limp \biat_i (\phi \land \psi)$.  The fourth item allows for viewing bringing it about as a modality, validating the rule of equivalents: if $\vdash \phi \leqv \psi$ then $\vdash \biat_i \phi \leqv \biat_i \psi$. (cf. Appendix B).\par

The language of \AILLs, $\mathcal{L}_{\AILLs}$ simply extends the definition of $ \mathcal{L}_{\ILLs}$, by adding a formula $\biat_C \phi$ for each coalition $C \in \textbf{C}$. 
The axioms of \AILLs are presented in Table 3. The Hilbert system is defined by extending the notion of deduction in Definition \ref{def:derivation-ill} by means of the new axioms in Table 3 and of two new rules for building deduction trees, cf. Definition 6.1.\par
A number of important differences are worth noticing, when discussing the principles of agency in linear logics. The first principle is captured by E1, that is, the linear version of T: $\biat_{C} \phi \implies \phi$. Since in linear logics all the tautologies are not provably equivalent, the second principle changes into an inference rule, that is ($\sim$ nec) in Definition 6.1: if $\vdash \phi$, then $\sim \biat_C \phi$. Note that this also entails that agents do not bring about contradictions. Moreover, the rule ($\biat_C$re) captures the fourth principle.\par
The principle for combining actions is crucial here: as we have discussed, it can be interpreted in substructural logic in a multiplicative way, by means of $\otimes$ (axiom E3), and in an additive way, by means of $\with$ (axiom E2).\footnote{The principles for combining actions have been criticised on the ground that coalitions $C$ and $D$ may have different goals, therefore it is not meaningful to view the action of $C \sqcup D$ as a joint action. However, the aggregative view of group actions defined by means of the majority rule presupposes that the group is not defined by means of a shared goal nor a shared intention. Therefore, Axioms E2 and E3 are legitimate from this point of view.}  

\begin{table}[t]
\begin{enumerate}\label{hlbiatC}
\item[-] All the axioms of \ILLD (cf. Table 1)
\item[E1] $\biat_C \phi \implies \phi$
\item[E2] $\biat_C \phi \wedge \biat_C \psi \implies \biat_C (\phi \wedge \psi)$
\item[E3] $\biat_C \phi \otimes \biat _D \psi \implies \biat_{C \sqcup D}  (\phi \otimes \psi)$
\end{enumerate}
\caption{Axioms of \AILLs}
\end{table}

\begin{definition}[Deduction in \AILLs]\label{def:derivation-biatC}
A \emph{deduction tree} $\mathcal{D}$ in \AILLs is inductively constructed as follows. (i)~The leaves of the tree are assumptions $\phi \t \phi $, for $\phi \in \mathcal{L}_{L \BIAT C}$, or $\t \psi$ where $\psi$ is an axiom in Table 3.\\
 (ii) If $\mathcal{D}$ and $\mathcal{D}'$ are deduction trees, then the trees in Definition \ref{def:derivation-ill} are also deduction trees in \AILLs. Moreover, the following are deduction trees (inductive steps). 

\begin{center}
\begin{tabular}{ccc}
$
 \AC{\stackrel{\mathcal{D}}{\vdash \phi \implies \psi}}
 \AC{\stackrel{\mathcal{D}'}{\vdash \psi \implies \phi}}
 \rl{$\biat_C$(re)}
 \BC{\t \biat_C \phi \implies \biat_C \psi}
 \dip
$
&
\hspace{10pt}
&
\AC{\t \phi}\rl{$\sim$ nec}
\UC{\t \sim \biat_C \phi}
\dip\\
\end{tabular}
\end{center}
\end{definition}

\subsection{Collective beliefs}

The logic for modelling collective beliefs is a minimal epistemic logic \cite{vardi1986epistemic}. Non-normal epistemic modalities have been applied also to modelling reasoning about epistemic states of strategic agents in \cite{parikh1995logical,bacharach1994epistemic}. By using beliefs operators to model propositional attitudes, we are viewing the content of beliefs as sharable among the community of agents, in accordance with the motivations of the framework of judgment aggregation.\footnote{Alternatively, we could have epistemic attitudes of agents by means of acceptance logic \cite{lorini2009logic}. This solution is the one adopted by \cite{BoellaEtAlCOIN2010}.}
The language of intuitionistic linear logic of beliefs with coalitions, label it by \BILLs , extends the language $\mathcal{L}_{\ILLs}$ by adding a doxastic modality $\bel_{C} \phi$.  The Hilbert system for \BILLs extends \ILLs with the following axioms (cf. \cite{vardi1986epistemic}). The notion of derivation extends that of \ILLs by adding the rule of equivalents for the modalities $\bel_{C} \phi$. 

\begin{table}[t]
\begin{enumerate}\label{helC}
\item[-] All the axioms of \ILLs (cf. Table 1)
\item[B1] $\sim \bel_{C} \bot$
\item[B2]  $\bel_C \phi \wedge \bel_C \psi \implies \bel_C (\phi \wedge \psi)$
\item[B3] $\bel_C \phi \otimes \bel _D \psi \implies \bel_{C \sqcup D}  (\phi \otimes \psi)$
\item[B4] $\bel_{C} \phi \implies \bel_{C} \bel_{C} \phi$
\item[B5] $\bel_{C} (\phi \wedge \psi) \implies \bel_{C} \phi$
\end{enumerate}
\caption{Axioms of \BILLs}
\end{table}

\begin{definition}[Deduction in \BILLs]\label{def:derivation-}
A \emph{deduction tree} in \BILLs  $\mathcal{D}$ is inductively constructed as follows. (i)~The leaves of the tree are assumptions $\phi \t \phi $, for $\phi \in \mathcal{L}_{\BILLs}$, or $\t \psi$ where $\psi$ is an axiom in Table 4 (base cases).\\
 (ii) If $\mathcal{D}$ and $\mathcal{D}'$ are deduction trees, then the trees in Definition \ref{def:derivation-ill} are also deduction trees in \AILLs. Moreover, the following are deduction trees (inductive steps). 
 
$$
 \AC{\stackrel{\mathcal{D}}{\vdash \phi \implies \psi}}
 \AC{\stackrel{\mathcal{D}'}{\vdash \psi \implies \phi}}
 \rl{$\bel_C$(re)}
 \BC{\t \bel_C \phi \implies \bel_C \psi}
 \dip
$$
 
\end{definition}

\noindent
Note that the rule ($\sim$ nec) that we used for actions is dropped for the case of epistemic logics, that is, we assume that agents may believe tautologies although we do not force them to do so. However, we assume that a modicum of rationality still applies and agents cannot believe contradictions, cf. axiom B1.
We also assume the positive introspection axiom, B4, which exhibits in this context the public nature of the belief. A negative introspection axiom can be defined as well, but we do not discuss it in this paper. Moreover, we assume that for the additive conjunction, agents can infer their belief in $\phi$ from the belief in $\phi \with \psi$ (axiom B5). 
The combination of beliefs is the original part here. Again, we distinguish between an additive and a multiplicative combination in order to keep track of the contribution of individuals to the acceptance of the collective belief. 

\subsection{Collective norms}

The use of non-normal modal logic to express deontic modalities was motivated in \cite{chellas80,goble2004proposal}. Moreover, non-normal deontic logics have been used to discuss institutional agency in \cite{carmo2001deontic} and to model weak permissions in \cite{roy2012logic}. Here we present a minimal version of deontic logic, for the sake of the exemplification, and we leave a proper treatment of deontic principles in \ILLs for future work. 
We extend the language of \ILLs by adding a number of modalities for obligations $\ob_{C}$, for $C \in \textbf{C}$. Moreover, we also assume the respective dual modalities (permission) $\per_{C}$. We term this system by $\OILLs$. 
The motivation for this formalisation is that we want to distinguish norms that constraint the behaviour of individuals, e.g. $\ob_{i} A$, to norms that applies to coalitions or collectives, e.g. $\ob_{\{1,3\}} \phi$.\footnote{For a taxonomy of norms that applies to groups, we refer to \cite{DignumCOIN2014}.}\par

The Hilbert system for \OILLs extends the case of \ILLs by adding the axioms in Table 5.

\begin{table}
\begin{enumerate}\label{hobC}
\item[-] All the axioms of \ILLs (cf. Table 1)
\item[O1] $\sim \ob_{C} \bot$
\item[O2]  $\ob_C \phi \wedge \ob_C \psi \implies \ob_C (\phi \wedge \psi)$
\item[O3] $\ob_C \phi \otimes \ob _D \psi \implies \ob_{C \sqcup D}  (\phi \otimes \psi)$
\item[O4] $\sim \ob_{C} \phi \implies  \per_{C} \sim \phi$
\item[O5] $\sim \per_{C} \phi \implies \ob_{C} \sim \phi$ 
\item[O6] $\ob_{C} \phi \implies \per_{C}\phi$
\end{enumerate}
\caption{Axioms of \OILLs}
\end{table}

\begin{definition}[Deduction in \OILLs]\label{def:derivationOILLs}
A \emph{deduction tree} $\mathcal{D}$  in \OILLs is inductively constructed as follows. (i)~The leaves of the tree are assumptions $\phi \t \phi $, for $A\phi\in \mathcal{L}_{\OILLs}$, or $\t \psi$ where $\psi$ is an axiom in Table 5 (base cases).\\
 (ii) I (ii) If $\mathcal{D}$ and $\mathcal{D}'$ are deduction trees, then the trees in Definition \ref{def:derivation-ill} are also deduction trees in \OILLs. Moreover, the following are deduction trees (inductive steps). 
 
$$
 \AC{\stackrel{\mathcal{D}}{\vdash \phi \implies \psi}}
 \AC{\stackrel{\mathcal{D}'}{\vdash \psi \implies \phi}}
 \rl{$\ob_C$(re)}
 \BC{\t \ob_C \phi \implies \ob_C \psi}
 \dip
$$
 \end{definition}

\section{A model of collective attitudes}

We are ready now to introduce the model of collective attitudes. Individual attitudes shall be captured by formulas $\Box_i \phi$, where $i$ is an individual (coalition) and $\Box$ is one of the modalities that we have encountered. Collective attitudes are then modelled by formulas $\Box_C \phi$, where $C$ is a winning coalition in a given profile with respect to the majority rule.\footnote{The majority rule is not going to be interpreted by means of any logic formula, as for instance in \cite[chapter 4]{Rubinstein2000}}
We focus in this paper on the majority rule, however the proposed model of collective attitudes can be applied in principle to any aggregation procedure.\par

In order to meet the standard hypothesis of JA, we assume that the individuals reason in classical logic. That is, we assume that the judgments of the individuals are expressed by means of the versions based on classical logic (\CL) of the logics $\AILLs$, $\BILLs$, and $\OILLs$, which are known from the literature;\footnote{Note that the classical versions of the modal logics that we have introduced can be obtained by adding (W) and (C) to the lists of axioms.} we label them by \ACL, \BCL, and \OCL, respectively. Let $\XCL$ be one of \ACL, \BCL, and \OCL.\par
All of this logics are defined, besides by the axioms of classical propositional logic, by means of the classical versions of the modal axioms and rules that we have encountered. For instance, each satisfies a rule of equivalents $\ob_{C}$(re), $\bel_{C}$(re), $\biat_{C}$(re). 
Observe that, the classical combination axioms, which dismiss the distinction between multiplicatives and additives, in this case have all the following form. Let $\Box_{C} \in \{\ob_C, \biat_C, \bel_C\}$ and $C, D \in \textbf{C}$:

\begin{equation}
\Box_{C} \phi \wedge \Box_{D} \psi \rightarrow \Box_{C \cup D} (\phi \wedge \psi)
\end{equation}

The \emph{agenda} $\Phi_{L}$ is a subset of $\mathcal{L}_{L}$ that is closed under negations. 
An individual judgment set $J_i$ is a subset of $\Phi_{L}$ such that it is consistent and complete wrt $L$; moreover, we assume that for every modal operator $\Box_{C}$, for $\Box \in \{\biat, \bel, \ob\}$, that occurs in formulas in $J_i$, $C=i$. That is, individual judgments are about individual beliefs/actions/obligations. This is motivated by the rationale of the framework of judgment aggregation, where individuals reason independently of each other on the matters of the agenda.\par
We want to associate the collective judgments accepted by majority with the coalitions of agents that support them. For this reason, we define the \emph{collective agenda} $\Phi^{C}_{L}$ that coincides with $\Phi_{L}$ for the formulas that do not contain modalities and for any formula in $\Phi_{L}$ that contains a modal operator, denoted by $\theta[\Box_{C} \psi]$, we have $\theta[\Box_{C} \psi] \in \Phi^{C}_{\XCL}$ for any $C$ such that $|C| \geq \frac{n+1}{2}$. That is, $\Phi^{C}_{L}$ contains all the possible collective attitudes that can be ascribed to winning coalitions of individuals.\par
We restate the definition of the majority rule as follows: $M_{C}: \mathcal{J}(\Phi_{\XCL})^{n} \to \mathcal{P}(\Phi^{C}_{L})$ is defined by

\begin{equation}
\phi \in M_{C}(\textbf{J})\; \text{iff}
\begin{cases}
|N_{\phi}| \geq \frac{n+1}{2},\;  \text{and $\phi$ does not contain modal operators}\\
\phi = \theta[\Box_{C} \psi]\; \text{and}\;  \theta[\Box_{1} \psi] \in J_{1},  \dots ,\theta[\Box_{l}\psi] \in J_{l},\; C = \{1, \dots, l\}\; \text{and}\; |C| \geq \frac{n+1}{2}
\end{cases}
\end{equation}

That is, for non-$\Box$ formulas the majority works as usual, for $\Box$-formulas, the majority aggregates them by merging the supporting individuals into the relevant winning coalition.\par
Note that, in case the codomain of $M_C$ is based on classical logic, $M_{C}(\textbf{J})$ may be inconsistent wrt to classical logic, therefore the codomain of $M_{C}$ is the power set of the agenda. For instance, take the following profile $\textbf{J}$ of sets of formulas of \BCL: 

\begin{center}
\begin{description}
\item $J_1 = \{\bel_1 \phi, \bel_1 \psi, \bel_1 (\phi \wedge \psi)\}$
\item $J_2 = \{\neg \bel_2 \phi, \bel_2 \psi,  \bel_2  \neg (\phi \wedge \psi)\}$
\item $J_3 = \{\bel_3 \phi, \neg \bel_3 \psi,  \bel_3 \neg (\phi \wedge \psi) \}$
\end{description}
\end{center}

\noindent
The majority rule $M_C(\textbf{J})$ returns the following set: $\{ \bel_{\{1,3\}} \phi, \bel_{\{1,2\}} \psi, \bel_{\{2,3\}} \neg (\phi \wedge \psi)\}$. 
We can show that $M_C(\textbf{J})$ is inconsistent in \BCL as follows. From $\bel_{\{1,3\}} \phi$ and $\bel_{\{1,2\}} \psi$, we can infer, by means of a classical combination axiom, $\bel_{\{1,2,3\}} (\phi \wedge \psi)$. From $\bel_{\{2,3\}} \neg (\phi \wedge \psi)$ and $\bel_{\{1,2,3\}} (\phi \wedge \psi)$, we can infer, again by a combination axiom, that $\bel_{\{1,2,3\}} (\phi \wedge \psi) \wedge \neg (\phi \wedge \psi)$, which entails $\bel_{\{1,2,3\}} \bot$, since $(\phi \wedge \psi) \wedge \neg (\phi \wedge \psi)$ is equivalent to $\bot$ and $\bel_{C}$ satisfies the classical version of axiom $\bel_{C}$(re).
Therefore, $M_C(\textbf{J})$ is inconsistent since $\bel_{\{1,2,3\}} \bot$ contradicts (a classical version of) axiom B1 $\neg \bel_{\{1,2,3\}} \bot$.\par

We assess the collective attitudes in $M_C(\textbf{J})$ wrt the logics \AILLs, \BILLs or \OILLs. We label by \XILLs one among \AILLs, \BILLs, and \OILLs. We extend the additive translation presented in Section \ref{sec:possibility} by $(\Box_{C}(\phi))' = \Box_{C}(\phi)'$. Hence, we can translate $M_{C}(\textbf{J})$ into its counterpart $M_{C}(\textbf{J})'$ in the additive fragment of the modal logics that we have defined. This returns an aggregation procedure $\mathcal{J}(\Phi_{\XCL})^{n} \to \mathcal{P}(\Phi^{C}_{\XILLs})$.  We can now assess the outcome of the majority rule defined on profiles based on classical modal logics wrt an agenda of propositions of \XILLs. We label the additive fragment of \XILLs by \aXILLs.\par

%

By means of Theorem \ref{th:aills}, the consistency result of the majority rule for \aILLs can be extended to \aXILLs. Note that the additive logics $\aAILLs$, $\aBILLs$, or $\aOILLs$ contain only the additive axioms for the modalities (i.e. they do not contain E3, B3 and O3).

\begin{theorem}\label{thm:axills}
For every agenda $\Phi_{\XCL}$ and profile $\bf{J}$, $M_{C}(\textbf{J})'$ is consistent and complete wrt $\aXILLs$.
\end{theorem}

\begin{proof}
Firstly, we show that also for $\aAILLs$, $\aBILLs$, or $\aOILLs$, it holds that every derivation $\Gamma \t \phi$ contains at most two formulas. For any $\box \in \{\biat, \bel, \ob \}$, notice that the rule $\Box$(re) does not increase the number of formulas in a derivation. Moreover, ($\sim$ nec) in $\aAILLs$ does not increase the number of formulas in a derivation as well. The additive versions of the axioms in $\aAILLs$, $\aBILLs$, or $\aOILLs$ cannot increase the number of formulas, since in the respective additive versions, we dropped the $\implies$-rule.\par
Therefore, by the argument of Theorem \ref{th:aills}, $M(\textbf{J})'$ is consistent and complete wrt modal \aXILLs. We show that also $M_{C}(\textbf{J})$ is consistent. Suppose by contradiction that it is not. Then, since every derivation $\Gamma \t \phi$ in $\aAILLs$, $\aBILLs$, or $\aOILLs$ contains at most two formulas, there exist $\phi_1$ and $\phi_2$ in $M_{C}(\textbf{J})$ such that $\phi_1, \phi_2 \t \bot $. Suppose $\phi_1 = \theta[\Box_{C} \phi]$ and  $\phi_2= \theta'[\Box_{C'}\psi]$, where $C$ possibly differs from $C'$. If $\phi_1$ and $\phi_2$ are in $M_{C}(\textbf{J})$, then there exist at least $\frac{n+1}{2}$ individuals such that  $\theta[\Box_{i} \phi] \in J_{i}$ and at least $\frac{n+1}{2}$ individuals such that  $\theta'[\Box_{i} \psi] \in J_{i}$. From this we infer that there exists an individual $h$ such that  $\theta[\Box_{h} \phi] \in J_{h}$ and  $\theta'[\Box_{h} \psi] \in J_{h}$. To conclude, we show that if $\theta[\Box_{C} \phi]$ and $\theta'[\Box_{C'}\psi]$ are inconsistent, then $\theta[\Box_{h} \phi] \in J_{h}$ and  $\theta'[\Box_{h} \psi] \in J_{h}$ are. If $C = C'$, then the argument holds because it amounts to a simple relabelling of the indexes of the modalities. If $C \neq C'$, then the proof of  $\theta[\Box_{C} \phi], \theta'[\Box_{C'}\psi] \t \bot$ may differ form the proof of  $\theta[\Box_{h} \phi], \theta'[\Box_{h} \psi]\t \bot$ only in that it may not use the instances of the axiom $ \Box_{X} \phi \with \Box_{X}\psi \rightarrow \Box_{X}(\phi \with \psi)$; therefore if one can prove $\bot$ in the case of $C \neq C'$, with less axioms, one can do that by identifying $C$ and $C'$ with $h$. Therefore, $J_{h}$ is inconsistent in \aXILLs, thus it is inconsistent in their classical counterparts.
\end{proof}

For instance, take the profiles of judgments sets in \BCL that we have previously encountered. The majority $M_{C}(\textbf{J})'$ returns now the collective set in \BILLs: $\{ \bel_{\{1,3\}} \phi, \bel_{\{1,2\}} \psi, \bel_{\{2,3\}} \sim (\phi \with \psi)\}$, which is \emph{not} inconsistent in \aBILLs. The reason is that we cannot derive $\bel_{C} (\phi \with \psi)$ and $\bel_{C} \sim (\phi \with \psi)$ for any $C$, from the formulas $\bel_{\{1,3\}} \phi$ and $\bel_{\{1,2\}} \psi$. This is due to the fact that the additive combination of beliefs is not applicable in this case.\par
By means of the additive connectives, we are only reasoning about propositions that are accepted because of a winning coalitions of agents; this means that the reasoning steps that are visible in  $M_{C}(\textbf{J})'$ as additive formulas are those that have already been performed autonomously by the individuals. In order to reason about propositions that are accepted by distinct winning coalitions and to model inferences performed at the collective level, we embed $M_{C}(\textbf{J})$ into the full logics \AILLs, \BILLs, and \OILLs.\par
A version of Fact \ref{fact} easily holds also for the modal counterparts, thus, we can define the aggregation procedure $Cl_{\XILLs}(M_{C}(\textbf{J})')$
that returns the deductive closure of $M_{C}(\textbf{J})'$ with respect to an given agenda $\Phi_{\XILLs}$ of formulas in \XILLs that contains $M_{C}(\textbf{J})'$ (cf. Equation \ref{eq:cl_maj}).\par
To add collective constraints, that are external to the individual agenda, we restate the definition of the aggregation procedure $T$ in this setting (cf. Definition \ref{def:twostep}). A set of \emph{collective constraints} is a set of formulas $\Sigma$ in $\mathcal{L}_{\XILLs}$ such that for each modal formula $\theta[\Box_{C}\phi] \in \Sigma$, $C$ is either a winning coalition of agents or $C = C_1 \sqcup \dots \sqcup C_l$, where each $C_i$ is a winning coalition of agents.\par 
We assume that, for each individual $i$, $i$ accepts a set of constraints $S_i \subseteq \Sigma$. Moreove, $S_i$ is compatible with the judgments of the individuals, once they are interpreted by means of the additive translation: $(J_i)' \cup S_i$ is consistent in \XILLs.

\begin{definition}\label{def:twostepC}
Let $\Phi_{\CL}$ an agenda in \CL and $\Phi_{\XILLs}$ an agenda in \XILLs that is composed of three disjoint sets: the additive translation of $\Phi_{\XCL}$, a set of constraints $\Sigma$, and a set of conclusion $\Delta$. The two step procedure $T_{C}: J(\Phi_{\XCL})^{n} \to \mathcal{P}(\Phi_{\XILLs})$ is defined as follows:

$$T(\textbf{J}) = \{\phi \mid \phi \in M_{C}(\textbf{J})'\; \text{or}\; \phi \in M(\textbf{S})\; \text{or}\; \phi \in \Delta\;\text{and}\; M_{C}(\textbf{J})', M(\textbf{S}) \t \phi \}$$
\end{definition}

\noindent
%
%
To guarantee the consistency of $T_{C}$ in case the constraints $\Sigma$ are expressed by means of any multiplicative formula, we have to assume  that every minimally inconsistent subset of $(\Phi_{\XCL})' \cup \Sigma$ has cardinality at most 2.

\subsection{Collective attitudes in the doctrinal paradox}\label{sec:cadp}

We conclude our analysis by exemplifying our treatment of collective attitudes for the case of the doctrinal paradox, cf. Section \ref{sec:doc}. We use the combination of the collective attitudes that we have introduced for the sake of illustration. In fact, in order to define a logic that combines all the logics \XILLs, we should extend the theory of \emph{fusion} of modal logics to the case of non-normal modalities and to substructural propositional axioms. This is in principle a viable solution, given that our modal logics are mutually independent. Let \UILLs be the fusion of the logics \XILLs. A version of Theorem \ref{thm:axills} for \UILLs holds. Note that, by putting all the modalities together, since no new inference rule nor axioms have been introduced, we still have that every derivation in such a system would contain at most two formulas. A detailed treatment of that goes beyond the scope of the present paper and it is left for future work.\par
We formalise the individual and collective propositional attitudes involved in the doctrinal paradox as follows. We use $C \in \textbf{C}$ as a variable for coalitions of agents. 

\begin{itemize}
\item[] $\bel_{C} p$ : $C$ believes that the confession was coerced.

\item[] $\bel_{C} q$ :  $C$ believes that the confession affected the outcome of the trial.

\item[] $\ob_{C} \biat_{C} r$ : $C$ is obliged to bring about that the trial is revised.
\end{itemize}

\noindent
The issues to be decided can be represented by means of the agenda of propositions in the fusion \UCL of the classical modal logics \ACL, \BCL, and \OCL:
 
$$\{\bel_{C} p, \bel_{C} q, \bel_{C}(p \wedge q), \bel_{C} \neg p, \bel_{C} \neg q, \bel_{C} \neg(p \wedge q),  B_{C}(p \wedge q) \rightarrow \ob_{C} \biat_{C} r,  \ob_{C} \biat_{C} r, \ob_{C} \neg \biat_{C} r\}$$

\noindent
The $C$s are here coalitions of agents in $N = \{1,2, 3\}$. We separate the individual agenda, for which $C = i$, for $i \in N$, and the collective agenda $\Phi^{C}_{L}$ where $C$ is any winning coalition wrt the majority rule.\par 
The legal doctrine is here represented by means of the combination of the deontic operator and of the action operator $\ob_{C} \biat_{C} r$ that means, intuitively, that under the condition that the court believes $\bel_{C} (p \wedge q)$, it is obligatory to bring about that the trial is revised, which captures the normative force of ``the trial \emph{must} be revised'' (see for instance \cite{carmo2001deontic}).\par
The critical part of the profile $\textbf{J}$ of the doctrinal paradox is represented in Table 6.

\begin{table}[h]\label{table:dpca}
\begin{tabular}{ccccccc}
 & $\bel_{C}p$ & $\bel_{C} (p \wedge q)$ & $\bel_{C} q$ & $\bel_{C}(p \wedge q) \rightarrow \ob_{C}\biat_{C} r$ & $\ob_{C}\biat_{C} r$ & $\ob_{C}\biat_{C} \neg r$\\
\hline
$1$ & 1 & 1 & 1 & 1 & 1 & 0\\   
$2$ & 1 & 0 & 0 & 1 & 0 & 1\\ 
$3$ & 0 & 0 & 1 & 1 & 0 & 1\\
\hline 
maj. & 1 & 0 & 1 & 1 & 0 & 1\\
\end{tabular}
\caption{Collective attitudes in the doctrinal paradox}
\end{table}

\noindent
The outcome of $M_{C}(\textbf{J})$ in classical logic on the profile above is then the following set of collective attitudes:

\begin{center}
$\{\bel_{\{1,2\}} p,\; \bel_{\{1,3\}} q,\; \bel_{\{2,3\}} \neg (p \wedge q),\; \bel_{\{1,2,3\}} (p  \wedge q) \rightarrow \ob_{\{1,2,3\}}\biat_{\{1,2,3\}} r,\; \ob_{\{2,3\}}\biat_{\{2,3\}} \neg r \}$
\end{center}

\noindent
The collective set is inconsistent wrt the (classically based) logic \UCL. It is illustrative to see why. From  $\bel_{\{1,2\}} p$ and $\bel_{\{1,3\}} q$ follows, by a classical versions of attitudes combination, that $\bel_{\{1,2,3\}} (p \wedge q)$, which entails, by \emph{modus ponens} and by $\bel_{\{1,2,3\}} (p  \wedge q) \rightarrow \ob_{\{1,2,3\}}\biat_{\{1,2,3\}} r$, the formula $\ob_{\{1,2,3\}}\biat_{\{1,2,3\}} r$.  We show that  $\ob_{\{1,2,3\}}\biat_{\{1,2,3\}} r$ contradicts $\ob_{\{2,3\}}\biat_{\{2,3\}} \neg r$  (which is in $M_{C}(\textbf{J})$) as follows. Again by classical combination, we obtain $\ob_{\{1,2,3\}}(\biat_{\{1,2,3\}} r \wedge \biat_{\{2,3\}} \neg r)$. From $\biat_{\{1,2,3\}} r$ and $\biat_{\{2,3\}} \neg r$ we infer, by means of a classical version of attitudes combination, that $\biat_{\{1,2,3\}}(r \wedge \neg r)$; $\biat_{\{1,2,3\}}(r \wedge \neg r)$ entails $\bot$, by a classical version of axiom $E1$ and by means of propositional reasoning that establishes that $(r \wedge \neg r)$ is equivalent to $\bot$. Therefore, $\biat_{\{1,2,3\}} r \wedge \biat_{\{2,3\}} \neg r \rightarrow \bot$ holds. On the other hand, by the \emph{ex falso quodlibet} principle of classical propositional reasoning, we have $\bot \rightarrow \biat_{\{1,2,3\}} r \wedge \biat_{\{2,3\}} \neg r$. Therefore, they are equivalent. By a classical version of the rule $\ob_{C}$(re), we obtain from $\ob_{\{1,2,3\}}(\biat_{\{1,2,3\}} r \wedge \biat_{\{2,3\}} \neg r)$ and from 
$\bot \leftrightarrow \biat_{\{1,2,3\}} r \wedge \biat_{\{2,3\}} \neg r$, $\ob_{\{1,2,3\}} \bot$, which contradicts axiom O1 (i.e. $\sim O_{C} \bot$).\par

By assessing the outcome of the majority rule $M_C(\textbf{J})$ in the additive fragment of \UILLs, i.e. by means of $M_C(\textbf{J})'$, we obtain the following set of formulas:

\[
\{\bel_{\{1,2\}}p, \bel_{\{1,3\}} q , \bel_{\{2,3\}} \sim (p \with q), \bel_{\{1,2,3\}} (p \with q) \rightsquigarrow \ob_{\{1,2,3\}} \biat_{\{1,2,3\}} r, \ob_{\{2,3\}}  \biat_{\{2,3\}} \sim r\}
\] 

In order to preserve the consistency of the majority rule, the formulas of the classical agenda are viewed, from the perspective of the collective reasoning, in additive terms, (cf. Theorem \ref{th:aills} and Definition \ref{def:twostep}). Our analysis of the doctrinal paradox entails that the solution provided by the aggregation procedure $M_{C}(\textbf{J})$' takes the conclusion-based horn of the dilemma, for which the conclusion $\ob_{\{2,3\}} \biat_{\{2,3\}} \sim r$ is accepted (prosaically, the trial must not be revised).\par

By embedding $M_{C}(\textbf{J})'$ into the full language of \UILLs, we infer from $\bel_{\{1,2\}}p$ and  $\bel_{\{2,3\}}q$ only the multiplicative combination $\bel_{\{1,2\} \sqcup \{2,3\}} (p \otimes q)$ (by means of Axiom B3), which cannot be used as a premise of the legal doctrine to infer that the trial must be revised. Consistently, the conclusion $\ob_{C} \biat_{C} r$ cannot be accepted for any combination of winning coalitions. $\ob_{C} \biat_{C} r$ is collectively accepted when a majority of individuals $i$ is autonomously accepting $\bel_{i} p$ and $\bel_{i}q$ and individually drawing the conclusion $\ob_{i} \biat_{i} r$ from $\bel_{i}p$ and $\bel_{i} q$. This corresponds, in turn, to a profile in which there is a single majority of agents that supports both $p$ and $q$. By contrast, in the profile of the doctrinal paradox, there is no single majority of individuals that support $p$ and $q$: $p$ is supported by $1$ and $2$, whereas $q$ is supported by $1$ and $3$. For that reason, $\ob_{C} \biat_{C}r$ is not accepted, and consistency is preserved. The conclusion-based solution is entailed by the fact that we view the legal doctrine in additive terms, that is, as a formula on which the individuals reason autonomously.\par
In order to get the premise-based solution---i.e., in order to infer $\ob_C \biat_C r$--- we need to drop the additive formulation of the legal doctrine and add the multiplicative version of it: $\bel_{C} (p \otimes q) \multimap \ob_{C} \biat_{C}r$. As we discussed, this constraint is external to the agenda on which the individuals cast their votes and it applies only to the collectively accepted propositions. We can obtain the premise-based reading by instantiating the procedure $T_{C}$ in the following way.
The individual agenda $\Phi_{\UCL}$ contains the following propositions $\{\bel_{C} p, \bel_{C} q, \bel_{C}(p \wedge q), \bel_{C} \neg p, \bel_{C} \neg q, \bel_{C} \neg(p \wedge q) \}$, where, $C \in \{1, 2, 3\}$. The collective agenda $\Phi_{\UCL}^{C}$ is then obtained by instantiating the formulas in the set above with any winning coalition $C \subseteq \{1, 2, 3\}$. The set of collective constraints is $\Sigma = \{\bel_{C \sqcup D} (p \otimes q) \implies \ob_{C \sqcup D} \biat_{C \sqcup D} r \}$, for any pair of winning coalitions $C$ and $D$ in $\{1,2,3 \}$.\footnote{Note that, as in the case of Section 5.1, by adding arbitrary constraints in the full language of \UILLs, consistency may be lost. However, in this case, the median property is again satisfied by the set of formulas $\Phi_{\UCL}' \cup \Sigma$.}
The set of conclusions is $\Delta = \{ \ob_{C \sqcup D} \biat_{C}r,  \ob_{C \sqcup D} \biat_{C \sqcup D} \sim r, \bel_{C \sqcup D} (p \otimes q)\}$, for any pair of winning coalitions $C$ and $D$ in $\{1,2,3 \}$. 
By assuming that the collective constraints are unanimously accepted, as it is standard in the view of doctrinal paradoxes,\footnote{Note that it is easy to adapt this treatment for the case of collective constraints that are a matter of aggregation.} the procedure $T_{C}$ returns the following set of formulas on the profile that coincides with that of Table 6 for the common part of the agenda:

 $$\{\bel_{\{1,2\}}p, \bel_{\{1,3\}} q , \bel_{\{2,3\}} \sim (p \with q), \bel_{C \sqcup D} (p \otimes q) \implies \ob_{C \sqcup D} \biat_{C \sqcup D} r, \ob_{\{1,2\} \sqcup \{1,3\}}  \biat_{\{1,2\} \sqcup \{2,3\}} \sim r \}$$
 
In this case, the conclusion  $\ob_{\{1,2\} \sqcup \{1,3\}}  \biat_{\{1,2\} \sqcup \{1,3\}} r$ is accepted (that is the trial must be revised). By means of $T_C$, we can also obtain the conclusion-based solution of the dilemma (that the trial must not be revised) by substituting the set of collective constraints $\Sigma$ with $\Sigma'$ containing the additive formulation of the legal doctrines $\bel_{C} (p \with q) \implies \ob_{C} \biat_{C} r$, for any winning coalition $C \subseteq \{1, 2, 3\}$.\par
We conclude by noticing, again, that the additive formulation is more demanding, as it requires that there exists a single winning coalition that accepts the premises of the doctrine.

\subsection{Towards corporate attitudes}

We have seen that it is in principle possible to model collective attitudes in a consistent way, provided that we keep track of the complex internal structure of collective agents, of the relationship between the coalitions, and of the distinction between individual and collective inferences. The view that we have discussed so far is an \emph{aggregative} view of collective attitudes and it can be applied to model a number of aspects of the view of \emph{common} collective attitudes, by replacing the majority rule with the unanimous aggregation. 
I did not approach a \emph{corporate} view of collective attitudes \cite{ListErkenntnis2014}. A corporate view requires the commitment to the existence of a \emph{single} agent $G$ who is the bearer of all the collective attitudes \cite{ListErkenntnis2014,PorelloEtAlFOIS2014}, (e.g. of the collectively accepted formulas in $T(\textbf{J})$), whereas in the previous modelling I have introduced a number of collective agents, one for each winning coalition that is responsible of supporting a certain proposition.
 The nature of the corporate agent $G$ can be approached in our framework by means of the following definition: \emph{if $\Box_{C} \phi$ is collectively accepted, for some coalition $C$, then $\Box_{G}\phi$} (e.g. if $\Box_{C} \phi$ it is in $T(\textbf{J})$, for some $C$, then $\Box_{G}\phi$). The reasoning principles that govern the logic of the corporate agent $G$, i.e. the principles of $\Box_{G}$, need to be carefully investigated, to preserve a modicum of rationality of $G$. In particular, they have to respect the structure of the coalitions that support $G$'s attitudes. For instance, we have to prevent the additive combination of attitudes to hold for $G$. Otherwise, we end up facing again the inconsistency manifested by the doctrinal paradox. In our previous example, from $\bel_{\{1,2\}} p$ we could infer $\bel_{G} p$ and from $\bel_{\{1,3\}}q$ would follow $\bel_G q$; by the additive combination, we would infer $\bel_G (p \with q)$ and that would contradict $\bel_{G} \sim (p \with q)$. By contrast, it is possible to prove that a version of the multiplicative axiom is harmless. We need to replace the disjoint union of coalitions with an operation of composition: $X \bullet Y = X \sqcup Y$, for $X, Y \neq G$ and $G \bullet X = G$. By rephrasing the multiplicative combination for $G$ and $\bullet$, we have $\Box_{G} p \otimes \Box_G q \implies \Box_{G \bullet G = G} (p \otimes q)$. Therefore, there is indeed a viable corporate view of collective attitudes based on the majoritarian aggregation: the definition of the collective attitudes of the corporate agent $G$ can be approached by means of the modalities $\Box_G$ that satisfy the axioms that we have presented for $\XILLs$ minus the additive combination. We leave the detailed treatment of $\Box_G$ and of its further principles for future work.

\section{Conclusion}
This paper has proposed a consistent modelling of collective attitudes, by focusing on the problematic case of collective attitudes that are defined by means of the majority rule. In order to preserve the consistency of the majority rule, a careful examination of the reasoning principles is needed; for that reason, the model of collective attitudes has been developed within substructural logics, that provide a fine-grained analysis of logical operators. We have presented three modal extensions of intuitionistic linear logic with strong negation to cope with action, beliefs, and obligations of individual and collective agents. The adequacy of this logics is shown by proving soundness and completeness wrt the suitable classes of models (Appendix B and C).\par
The substructural propositional base enables the consistency of the majority rule and allows for distinguishing inferences performed at the individual level and inferences performed at the collective level. The modal approach to individual and collective attitudes allows for making visible in the logical modelling the individual and the collective contributions to the ascription of collective attitudes. Therefore, this approach enables the formal understanding of collectivities as rational agents.\par
Future work concerns the extension of this treatment to other aggregation procedures. For instance, it is possible to extend the preservation of consistency exhibited by substructural logics to super-majoritarian aggregation procedures. On the other hand, it is possible to extend the propositional layer that we have used from linear to relevant logic (i.e. by adding contraction), in order to enrich the reasoning capabilities at the collective level, while still ensuring consistency. 
A number of problems remains open. Firstly, it is worth pursuing a systematic study of the principles of the logics obligations, beliefs, and action for the case of collective agents. Secondly, the study of the computational complexity of the fusion of the proposed logics deserves a dedicated treatment. Finally, it is worth studying the logical principles for modelling corporate attitudes.

\bibliographystyle{fundam}
\bibliography{groupattitudes}

\appendix
\section{Semantics of \ILLs}\label{app:sem}

A Kripke-like class of models for \ILL is substantially due to Urquhart \cite{UrquhartJSL1972}.  Here, we use the model for substructural logics developed in \cite{ono1985logics} with the addition of a strong negation \cite{kamide2006linear,wansing1993informational}. We term this structures \emph{Kripke resource algebras}.

\begin{definition}[Kripke resource algebra]
A \emph{Kripke resource algebra} is a \emph{semi-lattice ordered monoid} given by  $\langle M, \cap, \circ, e, \omega \rangle$, where $\langle M, \cap \rangle$ is a meet-semilattice with greatest element $\omega$ and the order $\leq$ is defined by $x \leq y$ iff $x \cap y = x$. $\langle M, \circ, e \rangle$ is a commutative monoid with identity $e$ that satisfies $x \circ \omega = \omega$  for any $x \in M$. Moreover, it holds that $z \circ (x \cap y) \circ w =  (z \circ x \circ w) \cap(z \circ y \circ w)$, for any $x,y,z,w \in M$. 

\end{definition}

To obtain a \emph{model} for \ILLs, two valuations on atoms $V^{+}: Atom \to 
\mathcal{P}(M)$ and $V^{-}: Atom \to \mathcal{P}(M)$ are added. The two valuations are introduced to model the strong negation: the positive valuation associates the verifiers of formulas in a model, the negative one their falsifiers \cite{wansing1993informational}.
They both have to satisfy the \emph{heredity} condition, which is written in this setting as follows: $m \cap n \in V^{+}(p)$ iff $m \in V^{+}(p)$ and $n \in V^{+}(p)$ and  $m \cap n \in V^{-}(p)$ iff $m \in V^{-}(p)$ and $n \in V^{-}(p)$. The conditions entail the version of heredity traditionally used in intuitionistic Kripke models, that is, if $m \in V^{+}(p)$ and $m \leq n$, then $n \in V^{+}(p)$.\par

 The truth conditions of the formulas of $\mathcal{L}_\ILLs$ are the following:

\begin{description} 
\item $m \models_\mathcal{M}^{+} p$ iff $m \in V^{+}(p)$.
\item $m \models_\mathcal{M}^{+} \textbf{1}$ iff $e \leq m$.
\item $m \models_{\mathcal{M}}^{+} \top$ for all $m \in M$.
\item $m \models_{\mathcal{M}}^{+} \bot$ iff  $m = \omega$.
\item $m \models_\mathcal{M}^{+} \phi \tensor \psi$ iff there exist $m_{1}$ and $m_{2}$ such that $m \geq m_{1} \circ m_{2}$ and $m_{1} \models_\mathcal{M} \phi$ and $m_{2} \models_\mathcal{M}^{+} \psi$.
\item $m \models_\mathcal{M}^{+} \phi \wedge \psi$ iff $m \models_\mathcal{M} \phi$ and $m \models_\mathcal{M} \psi$.
\item $m \models_\mathcal{M}^{+} \phi \oplus \psi$ iff  there exist $m_{1}$ and $m_{2}$ such that $m_{1} \cap m_{2} \leq m$ and $m_{1} \models^{+}_\mathcal{M} \phi$ or $m_{1} \models^{+}_\mathcal{M} \psi$ and $m_{2} \models^{+}_\mathcal{M} \phi$ or $m_{2} \models^{+}_\mathcal{M} \psi$.
\item $m \models_\mathcal{M}^{+} \phi \implies \psi$ iff for all $n \in M$, if $n \models_\mathcal{M} \phi$, then $n \circ m \models_\mathcal{M} \psi$.  
\item $m \models_\mathcal{M}^{+} \sim \phi$ iff $m \models_{\mathcal{M}}^{-} \phi$
\item $m \models_\mathcal{M}^{-} \sim \phi$ iff $m \models_\mathcal{M}^{+} \phi$
\item $m \models_\mathcal{M}^{-} p$ iff $m \in V^{-}(p)$.
\item $m \models_\mathcal{M}^{-} \textbf{1}$ iff $m =\omega$.
\item $m \models_\mathcal{M}^{-} \top$ iff $m =\omega$.
\item $m \models_\mathcal{M}^{-} \bot$ for all $x\in M$ .
\item $m \models_\mathcal{M}^{-} \phi \tensor \psi$ iff there exist $m_{1}$ and $m_{2}$ such that $m \geq m_{1} \circ m_{2}$ and $m_{1} \models^{-}_\mathcal{M} \phi$ and $m_{2} \models^{-}_\mathcal{M} \psi$.
\item $m \models_\mathcal{M}^{-} \phi \wedge \psi$ iff there exist $m_{1}$ and $m_{2}$ such that $m_{1} \cap m_{2} \leq m$ and $m_{1} \models^{-}_\mathcal{M} \phi$ or $m_{1} \models^{-}_\mathcal{M} \psi$ and $m_{2} \models^{-}_\mathcal{M} \phi$ or $m_{2} \models^{-}_\mathcal{M} \psi$.
\item $m \models_\mathcal{M}^{-} \phi \oplus \psi$ iff $m\models^{-}_\mathcal{M} \phi$ and $m \models^{-}_\mathcal{M}\psi$.
\item $m \models_\mathcal{M}^{-} \phi \implies \psi$ iff  there exist $m_{1}$, $m_{2}$ with $m_{1} \circ m_{2} \leq m$, such that $m_{1} \models^{+}_\mathcal{M} \phi$ and $m_{2} \models^{-}_\mathcal{M} \psi$.
\end{description}

Denote by $||\phi||^{+}_\mathcal{M}$  the set of worlds of $\mathcal{M}$ that verify $\phi$, $\{m \in M \mid m \models^{+} \phi\}$ and by $||\phi||^{-}_{\mathcal{M}}$ the set of falsifiers of $\phi$, $\{m \in M \mid m \models^{-} \phi\}$. A formula $\phi$ is \emph{true} in a model $\mathcal{M}$ if $e
\models^{+}_\mathcal{M} \phi$.\footnote{When the context is clear we will write $||\phi||^{+}$ instead of $||\phi||^{+}_\mathcal{M}$, and $m \models^{+} \phi$ instead of $m \models^{+}_\mathcal{M} \phi$.} A formula $\phi$ is \emph{valid} in Kripke resource algebra, denoted by $\models \phi$, iff it is true in every model. The heredity conditions easily extends by induction to every formula of \ILLs. By means of this semantics, \ILLs is sound and complete wrt to the class of Kripke resource algebras \cite{kamide2006linear}.


\section{A minimal modal extension of \ILLs}\label{app:mills}

We extend \ILLs by adding a non-normal modality $\Box$. In this section, we study the minimal non-normal modal operator $\Box$, while in the next sections we refine the modality to cope with propositional attitudes. The Hilbert system is a simple extension of the propositional case, we label this logic by \MILLs, (modal intuitionistic linear logic with strong negation).

\begin{definition}[Deduction in \MILLs]\label{def:derivation-biatC}
A \emph{deduction tree} $\mathcal{D}$ in \MILLs is inductively constructed as follows. (i)~The leaves of the tree are assumptions $\phi \t \phi $, for $\phi \in \mathcal{L}_{\MILLs}$, or $\t \psi$ where $\psi$ is an axiom in Table \ref{hills} (base cases).\\
 (ii) If $\mathcal{D}$ and $\mathcal{D}'$ are deduction trees, then the trees in Definition \ref{def:derivation-ill} are also deduction trees in \MILLs. Moreover, the following are deduction trees (inductive steps). 

\begin{center}
 \AC{\stackrel{\mathcal{D}}{\vdash \phi \implies \psi}}
 \AC{\stackrel{\mathcal{D}'}{\vdash \psi \implies \phi}}
 \rl{$\Box$(re)}
 \BC{\t \biat_C \phi \implies \biat_C \psi}
 \dip
\end{center}
\end{definition}

\noindent
The extension wrt \ILLs is provided by the rule $\Box$(re) that imposes that the modality is well-defined, that is, it preserves logically equivalent formulas (for the rule $\Box$(re), see\cite{PorelloTroquardJANCL2015}).\par

The models of \MILLs are defined by mean of a neighborhood semantics \cite{chellas80} defined on top of the Kripke resource algebra. A neighborhood function is a mapping $N: M \to \mathcal{P}(\mathcal{P}(M))$ that associates a
world $m$ with a set of sets of worlds (see~\cite{chellas80}). The intuitive meaning of the neighborhood in this setting is that it associates to each world a set of propositions that are \emph{verifiable necessities} at $m$. It is interesting to notice that to cope with a strong constructive negation for $\Box$-formulas, and with the two valuations, we need also to introduce a \emph{negative} neighborhood function that specifies which propositions are (constructively) falsifiable necessity at $m$. 
%

We require the neighborhood functions to satisfy the following conditions, that are needed in order to preserve heredity for modal formulas.

\begin{equation}
\label{princ:heredity1}
 X \in N^{+}(m \cap n) \text{ iff } X \in N^{+}(m) \text{ and } X \in N^{+}(n)
\end{equation}

\begin{equation}\label{princ:heredity2}
 X \in N^{-}(m \cap n) \text{ iff } X \in N^{-}(m) \text{ and } X \in N^{-}(n)
\end{equation}

The structure required to interpret \MILLs is then the following.

\begin{definition}[Modal Kripke resource algebra]
A \emph{modal Kripke resource algebra} is given by  $\langle M, \cap, \circ, e, \omega, N^{+}, N^{-} \rangle$, where  $\langle M, \cap, \circ, e, \omega \rangle$ is a Kripke resource algebra and $N^{+}$ and $N^{-}$ are neighborhood functions that satisfy (\ref{princ:heredity1}) and (\ref{princ:heredity2}).
\end{definition}

\noindent
Given a Kripke resource algebra, a model is specified by adding two valuations $V^{+}$ and $V^{-}$ that are assumed to satisfy heredity on atoms. In order to interpret the modalities in a modal Kripke resource algebra, we extend the semantic definitions for \ILLs by adding the following clauses:

\begin{itemize}
\item[] $m \models^{+} \Box \phi\; \text{iff}\; || \phi ||^{+} \in N^{+}(m)$
\item[] $m \models^{-} \Box  \phi\; \text{iff}\; || \phi ||^{+} \in N^{-}(m)$
\end{itemize}

\noindent
It is easy to show that heredity can be extended to modal formulas by means of conditions (\ref{princ:heredity1}) and (\ref{princ:heredity2}).\par
We say that a formula $\phi$ is true in a model iff $e \models^{+} \phi$. Moreover, $\phi$ is valid in the class of modal Kripke resource algebras iff $\phi$ is true for every model. We say that $\Gamma$ logically entails $\phi$, $\Gamma \models \phi$ iff for every model, if $e \models^{+} \Gamma$, then $e \models^{+} \phi$.\par

We show now that the semantic is adequate for \MILLs. The proof of soundness of \MILLs wrt the class of modal Kripke resource algebras simply extends the one for the propositional case \cite{kamide2006linear} to the modal rule.

\begin{theorem}[Soundness of \MILLs] If $\Gamma \t \phi$ is derivable in \MILLs, then $\Gamma \models \phi$. 
\end{theorem}

\begin{proof}
We only show the case of rule $\Box$(re). Assume $e \models^{+} \phi \implies \psi$ and $e \models^{+} \psi \implies \phi$, that is for all $m \in M$,  $m \models^{+} \phi$, iff $m \models^{+} \psi$. Thus, $||\phi ||^{+} = ||\psi||^{+}$. Hence, by the semantic clause for $\Box$-formulas, for every $m$, $||\phi||^{+} \in N(m)$ iff $||\psi||^{+} \in N(m)$, which entails $e \models^{+} \Box \phi \implies \Box \psi$.
\end{proof}

We turn now to show completeness. Firstly, we define a syntactical model, which is adapted for the propositional part from \cite{wansing1993informational} and \cite{kamide2006linear}, and we show that it is a canonical model.  Denote by $\Gamma$, $\Delta$ multisets of occurrences of formulas and by $\Delta^{\star} = \phi_1 \otimes \dots \otimes \phi_m$, for $\Delta = \{\phi_{1}, \dots, \phi_{m}\}$. We denote by $|\phi|$ the equivalence class of formulas modulo the interderivability relation in \MILLs: $\dashv$ $\vdash$. We denote by $V^{+c}(\phi)$ and $V^{-c}(\phi)$ the positive and negative extensions of a formula $\phi$ in $M^{c}$. Moreover, we denote by $[\phi] = \{|\Gamma^{\star}| \mid \Gamma^{\star} \t \Box \phi\}$.

\begin{definition}\label{def:mc}
Let $\mathcal{M}^c = \langle M^c,  e^c,  \omega^{c}, \circ^c, \cap^{c}, N^c, V^{+c}, V^{-c} \rangle$ such that:
\begin{itemize}
\item $M^c = \{ |\Gamma^{\star}| \mid \Gamma \text{ is a finite multiset of formulas} \}$;
\item $|\Gamma^{\star} | \circ^c |\Delta^{\star}| = |\Gamma^{\star} \otimes \Delta^{\star}|$;
\item $|\Gamma^{\star}| \cap |\Delta^{\star}| =  |\Gamma^{\star} \oplus \Delta^{\star}|$;
\item $e^c = |\textbf{1}|$;
\item $\omega^{c} = |\bot|$;
\item $|\Gamma^{\star}| \in V^{+c}(p)$ iff $\Gamma \vdash p$;
\item $|\Gamma^{\star}| \in V^{-c}(p)$ iff $\Gamma \vdash \sim p$;
\item $N^{+c}(| \Gamma^{\star} |) = \{ [\phi] \subseteq M^{c} \mid \Gamma^{\star} \vdash \Box \phi\}$.
\item $N^{-c}(| \Gamma^{\star} |) = \{ [\phi] \subseteq M^{c} \mid \Gamma^{\star} \vdash \sim \Box  \phi\}$
\end{itemize}
\end{definition}

For the propositional part, we refer to \cite{wansing1993informational}, that is $M^{c}$ is a Kripke resource algebra. We only show that $N^{+c}$  and $N^{-c}$ are neighborhood functions that satisfy conditions \ref{princ:heredity1} and \ref{princ:heredity2}.

\begin{lemma} $M^{c}$ is a modal Kripke resource algebra that satisfies \ref{princ:heredity1} and \ref{princ:heredity2}.
\end{lemma}

\begin{proof} 
We show that $N^{+c}$ satisfies conditions (\ref{princ:heredity1}). For (\ref{princ:heredity1}), assume that 
$X \in N^{+c}(|\Gamma^{\star}| \cap |\Delta^{\star}|)$, that is $X \in N^{+c}(|\Gamma^{\star} \oplus \Delta^{\star}|)$. Then, for $X = [\phi]$,
$\Gamma^{\star} \oplus \Delta^{\star} \t \Box \phi$ which is equivalent, by the inversion principle that holds for \ILLs \cite{wansing1993informational}, to having $\Gamma^{\star} \t \Box \phi$ and $ \Delta^{\star} \t \Box \phi$. The last two statements entail that $X = [\phi] \in N^{+c}(|\Gamma^{\star}|)$ and  $X = [\phi] \in N^{+c}(|\Delta^{\star}|)$. The case of condition (\ref{princ:heredity2}) is analogous: Suppose $X \in N^{-c}(\Gamma^{\star} \cap \Delta^{\star})$, then we have $X \in N^{-c}(|\Gamma^{\star} \oplus \Delta^{\star}|)$, which entails that, for $X = [\phi]$, $\Gamma^{\star} \t \neg \Box \phi$. From this, we conclude as before. 
\end{proof}

We can adapt the proof of the truth lemma of \cite{PorelloTroquardJANCL2015} and \cite{wansing1993informational} as follows. Let $\models^{+}_{c}$ the verification relation in the canonical model $M^{c}$ and $\models^{-}_{c}$ the falsification relation in the canonical model.

\begin{lemma}[Truth lemma]
$|\Gamma^{\star}| \models_{c}^{+} \phi$ iff $\Gamma \t \phi$ and $|\Gamma^{\star}| \models_{c}^{-} \phi$ iff  $\Gamma \t \sim \phi$
\end{lemma}

\begin{proof} The proof is as usual an induction on the complexity of the formula $\phi$.
We show the case of $\phi = \Box \psi$ for $\models^{+}_{c}$.  The induction hypothesis is: $|\Gamma^{\star}| \models_{c}^{+} \psi$ iff $\Gamma \t \psi$. We have that $|\Gamma^{\star}| \models_{c}^{+} \Box \psi$ iff $||\psi||_{M^{c}}^{+} \in N^{+c}(|\Gamma^{\star}|)$ (by the semantic clause for $\Box \psi$) iff  $\{|\Delta^{\star}| \mid |\Delta^{\star}| \models_{c}^{+} \psi\}  \in N^{+c}(|\Gamma^{\star}|)$ (by definition of $||.||_{M^{c}}^{+}$) iff 
$\{|\Delta^{\star}| \mid \Delta \t \psi)\}  \in N^{+c}(|\Gamma^{\star}|)$ (by the induction hypothesis) iff $[\psi] \in  N^{+c}(|\Gamma^{\star}|)$ (by definition of $[\psi]$) iff $\Gamma \t \Box \psi$ (by definition of $N^{+c}$).\par
\noindent
The case of  $\phi = \Box \psi$ for $\models^{-}_{c}$ is established as follows.  We have that $|\Gamma^{\star}| \models_{c}^{-} \Box \psi$ iff $||\psi||_{M^{c}}^{+} \in N^{-c}(|\Gamma^{\star}|)$ (by the semantic clause for $\Box \psi$) iff  $\{|\Delta^{\star}| \mid |\Delta^{\star}| \models_{c}^{+} \psi\}  \in N^{-c}(|\Gamma^{\star}|)$ (by definition of $||.||_{M^{c}}^{+}$) iff 
$\{|\Delta^{\star}| \mid \Delta \t \psi)\}  \in N^{-c}(|\Gamma^{\star}|)$ (by the induction hypothesis) iff $[\psi] \in  N^{-c}(|\Gamma^{\star}|)$ (by definition of $[\psi]$) iff $\Gamma \t \sim \Box \psi$ (by definition of $N^{-c}$).
\end{proof}

By means of the usual arguments, we can establish completeness of \MILLs. 

\begin{theorem}[Completeness of \MILLs] \MILLs is complete wrt the class of modal Kripke resource algebras that satisfy \ref{princ:heredity1} and \ref{princ:heredity2}: If $\Gamma \models \phi$, then $\Gamma \t \phi$. 
\end{theorem}

\begin{proof}
Suppose $\Gamma \nvdash \phi$, then by the truth lemma $|\Gamma^{\star}| \not\models_{c}^{+} \phi$, thus $\Gamma \not\models \phi$.  
\end{proof}

The logic \MILLs is the weakest non-normal modal logic that we can define on top of our propositional system. In the next sections, we use \MILLs as a base logic to develop modal logics for representing and reasoning about propositional attitudes.

\section{Semantics of the modal logics of actions, beliefs, and obligations}\label{app:abo}

\subsection{Semantics of \AILLs}\label{sec:linearbiatc}

The semantics of the bringing-it-about modality is defined by adding a neighborhood semantics on top of the Kripke resource algebra. The meaning of the neighborhood in this setting is that it associates to each world a set of propositions that can be elected by coalition $C$.  We take one neighborhood function $N_C$ for every coalition $C \in \textbf{C}$ and we define:

\begin{itemize}
\item[] $m \models^{+} \biat_C \phi\; \text{iff}\; || \phi ||^{+} \in N_{C}^{+}(m)$
\item[] $m \models^{-} \biat_C \phi\; \text{iff}\; || \phi ||^{+} \in N_{C}^{-}(m)$
\end{itemize}
The neighborhood  functions are assumed to satisfy conditions (\ref{princ:heredity1}) and (\ref{princ:heredity2}) in order to preserve heredity.
The rule ($\biat_{C}$re) does not require any further condition on the Kripke resource algebra, it is already true because of the definition of $\biat_C$, as we have seen in the previous section.\par The rule ($\sim$ nec) requires: 

\begin{equation}\label{eq:cond-no1}
\text{if}\; e \in X, \text{ then } X \in N_C^{-}(e)\; 
\end{equation}

Axiom E1 requires: 

\begin{equation}\label{eq:reflexivity}
\text{if }\; X \in N_C^{+}(m)\;\text{then}\; m \in X 
\end{equation}

We turn now to action compositions. E2 requires:

\begin{equation}{\label{eq:with}}
\text{if }\; X \in N_C^{+}(m)\; \text{and}\; Y \in N_{C}^{+}(m),\; \text{then}\; X \cap Y \in N_C^{+}(m) 
\end{equation}

Let $X \circ Y = \{x \circ y \mid x \in X\; \text{and}\; y \in Y\}$, the condition corresponding to the multiplicative version of action combination, Axiom E3, requires that the upper closure of $X \circ Y$, denote it by $(X \circ Y)^{\uparrow}$, is in $N_{C \sqcup D}(x \circ y)$:
  
\begin{equation}\label{cond:tensor}
\text{if}\; X \in N_{C}^{+}(m)\; \text{and}\;  Y \in N_{D}^{+}(n)\; \text{, then}\;  (X \circ Y)^{\uparrow} \in N^{+}_{C \sqcup D}(m \circ n)
\end{equation}

Summing up, \AILLs is evaluated over the Kripke resource algebras that satisfy  (\ref{princ:heredity1}), (\ref{princ:heredity2}),  (\ref{eq:cond-no1}), (\ref{eq:reflexivity}), (\ref{eq:with}), and (\ref{cond:tensor}).\par 


We approach now the proof of soundness and completeness of \AILLs wrt the suitable class of Kripke resource algebra. A proof of soundness and completeness for a modal logic of action based on \ILL in case of as single agent is provided in \cite{PorelloTroquardJANCL2015}. 

\begin{theorem}[Soundness of \AILLs]\label{thm:soundnessBIAT}
\AILLs is sound wrt the class of Kripke resource algebras that satisfy (\ref{princ:heredity1}) and (\ref{princ:heredity2}), (\ref{eq:cond-no1}), (\ref{eq:reflexivity}), (\ref{eq:with}), and (\ref{cond:tensor}): if $\Gamma \t \phi$, then $\Gamma \models \phi$. 
\end{theorem}

\begin{proof} The propositional part is proven in \cite{kamide2006linear,wansing1993informational}. For the modal part, we only present the cases for the rule ($\sim$ nec) and axioms E2 and E3. The other cases are handled in similar way in \cite{PorelloTroquardJANCL2015}.\\
\noindent
The rule ($\sim$ nec) is sound. Suppose $e \models^{+} \phi$, then  $e \in ||\phi||^{+}$. By condition (\ref{eq:cond-no1}), we have that $||\phi||^{+} \in N_{C}^{-}(e)$. Thus, by definition of $\models^{-}$, we have that $e \models^{-} \biat_{C} \phi$, that is $e \models^{+} \sim \biat_{C} \phi$.

\noindent
We show that axiom E2 is valid. That is, for every model, $e \models^{+} \biat_C \phi  \wedge  \biat_C \psi \implies \biat_C (\phi \wedge \psi)$. That means, by definition of $\implies$, for every $x$, if $x \models^{+} \biat_C\phi \wedge \biat_C \psi$, then $x \models^{+} \biat_C (\phi \wedge \psi)$. If $x \models^{+} \biat_C \phi \wedge \biat_C \psi$, then $x \models^{+} \biat_C \phi$ and $x \models^{+} \biat_C \psi$, that entails, by definition of $\biat_C$, that $||\phi||^{+} \in N_C^{+}(x)$ and $|\psi||^{+} \in N_C^{+}(x)$. Thus, by condition (\ref{eq:with}), we infer $||\phi||^{+} \cap ||\psi||^{+} \in N_C^{+}(x)$. That means $x \models^{+} \biat_C (\phi \wedge \psi)$.\\
\noindent
We show that axiom E3 is valid, $e \models^{+} \biat_C \phi \otimes \biat_D \psi \implies \biat_{C \sqcup D} (\phi \otimes \psi)$. That is, for every $x$, if $x \models^{+} \biat_C \phi \otimes \biat_D \psi$, then $x \models^{+} \biat_{C \sqcup D} (\phi \otimes \psi)$. If $x \models^{+}  \biat_C \phi \otimes \biat_D \psi$, then by definition of $\otimes$, there exist $y$ and $z$, such that $x \geq y \circ z$ and $y \models^{+} \biat_C \phi$ and $z \models^{+} \biat_D \psi$. Therefore, $||\phi||^{+} \in N_C^{+}(y)$ and $||\psi||^{+} \in N_D^{+}(z)$, this by condition (\ref{cond:tensor}), we infer that $(||\phi||^{+} \circ ||\psi||^{+})^{\uparrow} \in N^{+}_{C \sqcup D}(y \circ z)$. Thus, since $x \geq y \circ z$, by condition (\ref{cond:tensor}),  $(||\phi||^{+} \circ ||\psi||^{+})^{\uparrow} \in N^{+}_{C \sqcup D}(x)$, that is $x \models^{+} \biat_{C\sqcup D} (\phi \otimes \psi)$. 

\end{proof}

We adapt the proof of completeness to the case of \AILLs. The construction of the canonical model is analogous to the one presented in Definition \ref{def:mc}.

\begin{lemma}\label{lem:mc}
$\mathcal{M}^c$ is a modal Kripke resource algebra that satisfies (\ref{princ:heredity1}), (\ref{princ:heredity2}), (\ref{eq:cond-no1}), (\ref{eq:reflexivity}), (\ref{eq:with}), and (\ref{cond:tensor}).
\end{lemma}

\begin{proof}
We only show the case of condition (\ref{eq:cond-no1}), (\ref{eq:with}), and (\ref{cond:tensor}), which differs from the proof in \cite{PorelloTroquardJANCL2015}.\\
\noindent
For condition $(\ref{eq:cond-no1})$, suppose $e^{c} = |\textbf{1}| \in X$ and $X = [\phi]$. If $|\textbf{1}| \in [\phi]$, then, by definition, $\textbf{1} \t \phi$.  By the rule $(\sim)$nec, we have that $\t \sim \biat_{C} \phi$, thus $[\phi] \in N_{C}^{-c}(|\textbf{1}|)$.

\noindent
For condition (\ref{eq:with}), suppose $X \in N_C^{+c}(|\Gamma^{\star}|)$ and $Y \in N_{C}^{+c}(|\Gamma^{\star}|)$. By definition of $N_{C}^{+c}$, if $X = [\phi]$, $\Gamma \t \biat_C \phi$ is provable in the Hilbert system. Analogously, $\Gamma \t  \biat_C \psi$, where $Y = [\psi]$. Then, we can prove in the Hilbert system that $\Gamma \t \biat_C \phi \wedge \biat_C \psi$, by means of the $\wedge$-rule: 

$$
\AC{\stackrel{\mathcal{D}}{\Gamma \t \biat_C \phi}}
\AC{\stackrel{\mathcal{D}'}{\Gamma \t \biat_C \psi}} \rl{$\wedge$-rule}
\BC{\Gamma \t \biat_C \phi \wedge \biat_C \psi}
\dip
$$

By the $\implies$-rule (i.e. \emph{modus ponens}), we conclude $\Gamma \implies \biat_C (\phi \wedge \psi)$ as follows: 

$$
\AC{\Gamma \t \biat_C \phi \wedge \biat_C \psi}
\AC{\t \biat_C \phi \wedge \biat_C \psi \implies \biat_C (\phi \wedge \psi)}
\rl{$\implies$-rule}
\BC{\Gamma \t \biat_C (\phi \wedge \psi)}
\dip
$$

Since $\Gamma \t \biat_C (\phi \wedge \psi)$, we have that $[\phi \wedge \psi] \in N_{C}^{+c}(|\Gamma^{\star}|)$. Therefore, we can conclude since  
$[\phi \wedge \psi] = [\phi] \cap [\psi] = X \cap Y$.\\

\noindent
For condition (\ref{cond:tensor}), assume $X \in N_{C}^{+c}(|\Gamma^{\star}|)$,  $Y \in N_{D}^{+c}(|\Delta^{\star}|)$. By definition of canonical neighborhood, we have: $\Gamma \t \biat_C \phi$, $\Delta \t \biat_D \psi$, where $[\phi]= X$ and $[\psi] = Y$. We can prove that $\Gamma, \Delta \t \biat_C \phi \otimes \biat_D \psi$ as follows.

$$
\AC{\Gamma \t \biat_C \phi}
\AC{\t  \biat_C \phi \implies (\biat_D \psi \implies  (\biat_C \phi \otimes \biat_D \psi))\; \text{(ax. )}}
\rl{$\implies$-rule}
\BC{\Gamma \t \biat_D \psi \implies  \biat_C \phi \otimes \biat_D \psi}
\AC{\Delta \t \biat_D \psi}
\LeftLabel{$\implies$-rule}
\BC{\Gamma, \Delta \t \biat_C \phi \otimes \biat_C \psi}
\dip
$$

By means of axiom E3, we infer $\Gamma, \Delta \t \biat_{C \sqcup D} (\phi \otimes \psi)$.

$$
\AC{\Gamma, \Delta \t \biat_C \phi \otimes \biat_C \psi}
\AC{\t \biat_C \phi \otimes \biat_D \psi \implies \biat_{C \cup D} (\phi \otimes \psi)\; \text{(ax 13, $C \cap D = \emptyset$})}
\rl{$\implies$-rule}
\BC{\Gamma, \Delta \t \biat_{C \cup D} \phi \otimes \psi}
\dip
$$

Therefore, $[\phi \otimes \psi] \in N^{+c}_{C \sqcup D} (|\Gamma^{\star}| \circ |\Delta^{\star}|)$. We conclude by noticing that $(X \circ Y)^{\uparrow} = ([\phi] \circ [\psi])^{\uparrow}$.
\end{proof}

The proof of the truth lemma simply adapts our previous proof, thus completeness for \AILLs follows. 

\begin{theorem}[Completeness of \AILLs]
\AILLs is complete wrt the class of modal Kripke resource algebras that satisfy  (\ref{princ:heredity1}), (\ref{princ:heredity2}), (\ref{eq:cond-no1}), (\ref{eq:reflexivity}), (\ref{eq:with}), and (\ref{cond:tensor}): If $\Gamma \models \phi$, then $\Gamma \t phi$.
\end{theorem}

\subsection{Semantics of \BILLs}

The semantic of the epistemic modalities is defined by adding neighborhood functions to the Kripke resource algebra. To distinguish them from the neighbourhood functions for actions, we denote them by $B_{C}$, where $C$ is a symbol for a coalition in $\textbf{C}$. In this setting, the interpretation of neighbourhood is the one provided by \cite{vardi1986epistemic}, that is, the function associates to a world the set of propositions that are believed by the coalition $C$ at that world. Accordingly, we shall assume a number of neighborhood functions, labelled by $B_{C}$. The semantic definitions extend those of \ILLs by adding the following conditions.

\begin{itemize}
\item[] $m \models^{+} \bel_C \phi\; \text{iff}\; || \phi ||^{+} \in B^{+}_{C}(m)$
\item[] $m \models^{-} \bel_C \phi\; \text{iff}\; || \phi ||^{+} \in B^{-}_{C}(m)$
\end{itemize}

We present the semantic conditions that correspond to the axioms of \BILLs. Firstly, we ensure that wrt the new modalities, heredity is preserved by assuming the conditions (\ref{princ:heredity1}), (\ref{princ:heredity2}). The minimal rationality assumption encoded by axiom B1 corresponds to the following condition, which excludes that a coalition of agents holds contradictory beliefs. 

\begin{equation}
\label{cond:absurdbeliefs}
 \{\omega\} \in B^{-}_C(e)\
 \end{equation}

Recall that $\{\omega\}$ is $||\bot||^{+}$. The positive introspection axiom (B4) corresponds to the following condition.

\begin{equation}
\label{cond:posintro}
\text{if } X \in B_C^{+}(m) \text{ then } \{n \mid X \in B_{C}(n)\} \in B_{C}^{+}(m)
\end{equation}

%

The following condition adequately captures axiom B5.

\begin{equation}
\label{cond:withmonB}
\text{if } X \in B_C^{+}(m)\; \text{and}\;  Y \subseteq X\; \text{then}\; Y \in B^{+}_{C}(m)
\end{equation}

We turn now to beliefs compositions. The conditions are analogous to those for actions composition. Axiom B2 requires:

\begin{equation}{\label{cond:withB}}
\text{if }\; X \in B_C^{+}(w)\; \text{and}\; Y \in B_{C}^{+}(w),\; \text{then}\; X \cap Y \in B_C^{+}(w) 
\end{equation}

The multiplicative version of belief combination, axiom B3, require:
  
\begin{equation}\label{cond:tensorB}
\text{if}\; X \in B_{C}^{+}(x)\; \text{and}\;  Y \in B_{D}^{+}(y)\; \text{, then}\;  (X \circ Y)^{\uparrow} \in B_{C \sqcup D}^{+}(x \circ y)
\end{equation}

We can now establish soundness and completeness of \BILLs wrt the suitable class of modal Kripke resource algebras.

\begin{theorem}[Soundness of \BILLs]\label{soundnessEPI} \BILLs is sound wrt the class of modal Kripke resource algebras that satisfy  (\ref{princ:heredity1}), (\ref{princ:heredity2}), (\ref{cond:absurdbeliefs}), (\ref{cond:posintro}),  (\ref{cond:withmonB}), (\ref{cond:withB}), and (\ref{cond:tensorB}).
\end{theorem}

\begin{proof}
 We only show the case of Axioms B1, B4, and B5.
For B1, we have to show that  $e \models^{+} \sim \bel_{C} \bot$. That is $e \models^{-} \bel_{C} \bot$ which is equivalent by definition to $||\bot||^{+} \in B_{C}^{-}(e)$. Since $||\bot||^{+} = \{ \omega\}$, by condition (\ref{cond:absurdbeliefs}), we conclude.\par
For B4, suppose $m \models^{+} \bel_{C} \phi$, then $||\phi||^{+} \in B_{C}(m)$. By condition (\ref{cond:posintro}), $\{n \mid ||\phi|| \in B_{C}(n)\} \in B_{C}^{+}(m)$. Therefore, $||B_{C} \phi ||^{+} \in B_{C}^{+}(m)$, thus $m \models^{+} \bel_{C} \bel_{C} \phi$.\par
For B5, suppose $e \models^{+} \bel_{C}(\phi \wedge \psi)$, then $||\phi \wedge \psi||^{+} \in B_{c}(e)$. Since $||\phi \wedge \psi||^{+} \subseteq ||\phi||^{+}$, by condition (\ref{cond:withmonB}), we have that $||\phi||^{+} \in B_{C}^{+}(e)$, therefore $e \models^{+} \bel_{C} \phi$.
\end{proof}

For completeness, we adapt the construction of the canonical model $M^{c}$.

\begin{lemma}\label{lem:mc}

$\mathcal{M}^c$ is a modal Kripke resource algebra that satisfies  (\ref{princ:heredity1}), (\ref{princ:heredity2}), (\ref{cond:absurdbeliefs}), (\ref{cond:posintro}),  (\ref{cond:withmonB}), (\ref{cond:withB}), and (\ref{cond:tensorB}).
\end{lemma}

\begin{proof}
We only show the case of condition (\ref{cond:absurdbeliefs}) and (\ref{cond:posintro}), as the other arguments are very similar to those presented in the proof of Lemma \ref{lem:mc}. For condition (\ref{cond:absurdbeliefs}), we show that $\{w^{c}\} = \{|\bot|\} \in B_{C}^{-c}(|\textbf{1}|)$. Since $\t \sim B_{C}^{-}\bot$, $[\bot] = |\bot| \in B_{C}^{-c}(|\textbf{1}|)$.\par

For condition \ref{cond:posintro}, assume $X \in B_{C}^{+c}(\Gamma)$, we have to show that $\{ \Delta \mid X \in B_{C}(\Delta)\} \in B_{C}^{+c}(\Gamma)$. By assumption, $X \in \{ [ \phi ] \mid \Gamma \t \bel_{C} \phi \}$. The set $\{ \Delta \mid [\phi] \in B_{C}^{+c}(\Delta)\}$ coincides by definition with $\{ \Delta \mid \Delta \t \bel_{C} \phi\} = [\bel_{C} \phi]$.  By axiom B4, we infer $\Gamma \t \bel_{C} \bel_{C} \phi$, that is by definition  $[\bel_{C} \phi] \in B_{C}^{+c}(\Gamma)$.
\end{proof}

The proof of the truth lemma and of the completeness are easy adaptation of the previous proofs.

\begin{theorem}[Completeness \BILLs]\label{completenessEPI} \BILLs is complete wrt the class of Kripke resource algebra that satisfy (\ref{princ:heredity1}), (\ref{princ:heredity2}), (\ref{cond:absurdbeliefs}), (\ref{cond:posintro}),  (\ref{cond:withmonB}), (\ref{cond:withB}), and (\ref{cond:tensorB}).
\end{theorem}

\subsection{Semantics of \OILLs}

The semantics of collective obligations is obtained again by adding a number of neighborhood functions $O_{C}$. The semantics of the modal formulas of $\mathcal{L}_{\OILLs}$ is then as follows. 

\begin{itemize}
\item[] $m \models^{+} \ob_{C} \phi$ iff $||\phi||^{+} \in O_{C}^{+} (m)$
\item[] $m \models^{-} \ob_{C} \phi$ iff $||\phi||^{+} \in O_{C}^{-} (m)$

\item[] $m \models^{+} \per_{C} \phi$ iff $||\phi||^{-} \notin O_{C}^{-}(m)$
\item[] $m \models^{-} \per_{C} \phi$ iff $||\phi||^{-} \in O_{C}^{+}(m)$
\end{itemize}

Note that the duality between $\ob_{C}$ and $\per_{C}$ is ensured: $\ob_{C} \phi$ is equivalent to $\sim \per_{C} \sim \phi$. 
We have that $m \models^{+} \ob_{C} \phi$ iff $||\phi||^{+} \in O_{C}^{+}(m)$ iff $|| \sim \phi ||^{-} \in O_{C}^{+}(m)$ iff $m \models^{-} \per_{C} \sim \phi$ iff $m \models^{+} \sim \per_{C} \sim \phi$. Firstly, we assume that the neighborhood functions $O_{C}$ satisfy (\ref{princ:heredity1}) and (\ref{princ:heredity2}). 
The conditions required by the new axioms are then the following.  Axiom O1 requires, as in the case of beliefs:

\begin{equation}
\label{cond:trueob}
 \{\omega\} \in B_C^{-}(e)\
 \end{equation}

Axioms O2 and O3 require the familiar conditions on combinations of propositions.

%

\begin{equation}{\label{cond:withO}}
\text{if }\; X \in O_C^{+}(m)\; \text{and}\;  Y \in O_{C}^{+}(m),\; \text{then}\; X \cap Y \in O_C^{+}(m) 
\end{equation}
  
\begin{equation}\label{cond:tensorO}
\text{if}\;  X \in O_{C}^{+}(x)\;  \text{and}\;  Y \in O_{D}^{+}(y)\; \text{, then}\;  (X \circ Y)^{\uparrow} \in O_{C \sqcup D}^{+}(x \circ y)
\end{equation}

Axioms O4 and O5 are already valid. Finally, the situation of axiom O6 is delicate. In the classical case a condition such as $\text{If } X \in O_{C}^{+}(x), \text{ then } M\setminus X \in O_{C}^{-}(x)$ would suffice. Unfortunately, the set-theoretic complement is not available in the Kripke resource algebra, thus, we do not have the direct means to talk about the set of falsifiers in terms of set-theoretic construction. One way to cope with this problem is to introduce a new pairs of neighborhoods $P_{C}^{+}$ and $P_{C}^{-}$  and putting the constraint:

 \begin{equation}\label{cond:obper}
\text{if}\;  X \in O_{C}^{+}(x)\;  \text{, then}\;  X \in P_{C}^{+}(x \circ y)
\end{equation}

We leave a proper treatment of this strategy for a dedicated work. We can now establish soundness and completeness for \OILLs (minus axiom O6).

\begin{theorem}[Soundness of \OILLs]\label{soundnessO} \OILLs is sound wrt the class of modal Kripke resource algebras that satisfy  (\ref{princ:heredity1}), (\ref{princ:heredity2}),(\ref{cond:trueob}), (\ref{cond:withO}), and (\ref{cond:tensorO}). 
\end{theorem}


%

\begin{theorem}[Completeness \OILLs]\label{completenessEPI} \OILLs is complete wrt the class of Kripke resource algebra that satisfy (\ref{princ:heredity1}), (\ref{princ:heredity2}),(\ref{cond:trueob}), (\ref{cond:withO}), (\ref{cond:tensorO}). 
\end{theorem}

\end{document}